\documentclass[compsoc]{IEEEtran}
\usepackage{cite}
\usepackage{amsmath,amssymb,amsthm,amsfonts}
\usepackage{xfrac}
\usepackage{subcaption}
\usepackage{textcomp}
\usepackage{xcolor}
\def\BibTeX{{\rm B\kern-.05em{\sc i\kern-.025em b}\kern-.08em
    T\kern-.1667em\lower.7ex\hbox{E}\kern-.125emX}}

\usepackage[utf8]{inputenc}
\usepackage{algorithm}
\usepackage{algpseudocode}
\usepackage{multirow}
\usepackage{graphicx}
\usepackage{tabularx,booktabs}
\usepackage{textcomp}
\usepackage{xcolor}
\usepackage{enumerate}
\usepackage{enumitem} 
\usepackage{wrapfig}
\usepackage{tcolorbox}

\usepackage{times}
\usepackage{soul}
\usepackage{url}
\usepackage[hidelinks]{hyperref}
\usepackage[utf8]{inputenc}
\usepackage{caption}
\usepackage{booktabs}
\usepackage[absolute,overlay]{textpos}

\usepackage{nomencl}
\makenomenclature

\newtheorem*{theorem*}{Theorem}
\newtheorem*{corollary*}{Corollary}

\theoremstyle{remark}

\newtheorem*{claim*}{Claim}

\theoremstyle{plain}

\newtheorem{constraint}{Prior Knowledge}

\newcommand*\diff{\mathop{}\!\mathrm{d}}
\newcommand*\MyExp{{\mathbb E}}

\DeclareMathOperator*{\argmin}{arg\,min}
\DeclareMathOperator*{\argmax}{arg\,max}

\newcount\Comments  
\usepackage{color}
\definecolor{darkgreen}{rgb}{0,0.5,0}
\definecolor{purple}{rgb}{1,0,1}
\newcommand{\kibitz}[2]{\ifnum\Comments=1\textcolor{#1}{#2}\fi}
\newcommand{\kizito}[1]{\kibitz{red}      {[Kizito: #1]}}

\newcommand{\xingyu}[1]  {\kibitz{darkgreen}   {[Xingyu: #1]}}

\begin{document}
\begin{textblock*}{20cm}(2cm,1cm)
\textcolor{red}{\textbf{Accepted by IEEE Transactions on Software Engineering}}
\end{textblock*}


\title{The Unnecessity of Assuming Statistically Independent Tests in Bayesian Software Reliability Assessments}

\author{Kizito~Salako,
        Xingyu~Zhao
\IEEEcompsocitemizethanks{\IEEEcompsocthanksitem K.~Salako is with the Centre for Software Reliability, City, University of London, Northampton Square EC1V 0HB, U.K.
(email: k.o.salako@city.ac.uk)
\IEEEcompsocthanksitem X.~Zhao is with the Department of Computer Science, University of Liverpool, Ashton Street L69 3BX, U.K. (email: xingyu.zhao@liverpool.ac.uk)}
\thanks{Manuscript received August 15, 2022; revised December 21, 2022.}}



\IEEEtitleabstractindextext{%
\begin{abstract}
When assessing a software-based system, the results of Bayesian statistical inference on operational testing data can provide strong support for software reliability claims. For inference, this data (i.e. software successes and failures) is often assumed to arise in an independent, identically distributed (i.i.d.) manner. In this paper we show how conservative Bayesian approaches make this assumption unnecessary, by incorporating one's doubts about the assumption into the assessment. We derive conservative confidence bounds on a system's probability of failure on demand (\emph{pfd}), when operational testing reveals no failures. The generality and utility of the confidence bounds are illustrated in the assessment of a nuclear power-plant safety-protection system, under varying levels of skepticism about the i.i.d. assumption. The analysis suggests that the i.i.d. assumption can make Bayesian reliability assessments extremely optimistic -- such assessments do not explicitly account for how software can be very likely to exhibit no failures during extensive operational testing despite the software's \emph{pfd} being undesirably large. 
\xingyu{Also to highlight some practical insights, maybe based on what KS writes later in the papers, but in a shorter form of ``caution is warranted'': confidence based on i.i.d. assumptions can be very optimistic as failure-free tests accumulate; ii) some forms of doubt about the i.i.d. assumption significantly impact confidence, while other forms do not; iii) surprisingly, failure-free testing can eventually undermine confidence in the system satisfying the bound. }
\xingyu{In the NWTES paper, we have something like ``For instance, even if one is certain that some new software must be more reliable than an old product, using the reliability distribution for the old software as a prior distribution when assessing the new system gives optimistic, not conservative, predictions for the posterior reliability of the new system after seeing operational testing evidence.''}
\end{abstract}

}
\maketitle


\IEEEpeerreviewmaketitle

\IEEEraisesectionheading{\section{Introduction}\label{sec_intro}}

\IEEEPARstart{C}{onsider} a software-based on-demand system subjected to black-box operational testing. During testing, an assessor observes the system -- in particular, the software --  as it responds to each demand in a random sequence of demands. By noting those demands the software correctly responds to, and those it fails on, the assessor intends to gain enough confidence to support claims of the system being sufficiently reliable\footnote{This work focuses on software reliability; i.e. only software failures are considered to cause system failure.}. For example, confidence in the system's unknown \emph{probability of failure on demand} (\emph{pfd}) -- $X$ say -- being sufficiently small. A Bayesian approach to gaining such confidence typically requires 2 things of our assessor: {\bf i)} \emph{prior} to operational testing, the assessor must scrutinise all evidence related to the system's operational readiness. Via such probing and the assessor's domain expertise, the assessor forms beliefs about which ranges of \emph{pfd} values are most likely, and which ranges are less so. Examples of evidence include formal analyses of a codebase, the performance of a system during operation \cite{thomas_e_wierman_reliability_2001}, the historical performance of similar systems \cite{bunea_two-stage_2005,porn_two_stage_1996}, and improvements in software development approaches \cite{fenton_software_1999}; {\bf ii)} the assessor must postulate a \emph{statistical model} -- in essence, a family of stochastic processes, any of which could characterise the occurrence of the system's successes and failures during operation. These processes should exhibit statistical properties consistent with the assessor's beliefs.

For the statistical model it is often assumed that software failures and successes are the outcomes of \emph{independent, identically distributed} (i.i.d.) Bernoulli trials. This is mathematically convenient and guarantees the \emph{strong law of large numbers} -- i.e., operational testing statistics converge \emph{almost surely} to dependability measures of interest (e.g. \emph{pfd}). The i.i.d. assumption can also be reasonable on practical grounds: a typical justification is when demands are rare, and the states/properties of the software and its operational environment are effectively ``reset'' inbetween these rare demand occurrences. Nevertheless, we must point out that \emph{any} statistical model used in reliability assessment -- including one reliant on the i.i.d. assumption -- is necessarily a postulate. One does not (cannot?) know with complete certainty that such model assumptions are valid in practice. The skeptical assessor allows for the possibility that the i.i.d. assumption \emph{does not} hold, even if this is unlikely. By incorporating their skepticism into the assessment, our assessor can investigate whether their doubts have a significant impact on their confidence in the system. In the best case, their confidence is insensitive to significant departures from the i.i.d. assumption. But in the worst case, ignoring the slightest failure correlations could lead to seriously misplaced confidence and dangerously optimistic reliability claims.

To aid the skeptical assessor, this work presents a new application of \emph{conservative Bayesian inference} (CBI) techniques for reliability assessments. The paper's research contributions are:
\begin{enumerate}[label={\arabic*)}]
\item incorporating a formal notion of ``\emph{doubting}'' the i.i.d. assumption into reliability assessments (Section~\ref{sec_consconfbounds});

   
\item a novel CBI technique that accounts for correlated operational testing outcomes (Sections~\ref{sec_KlotzModel}, \ref{sec_consconfbounds} and Appendices~\ref{sec_app_C}, \ref{sec_app_E})\xingyu{KS polish later}; 
    
\item  a theorem that gives conservative posterior confidence bounds on a system's \emph{pfd} (Section~\ref{sec_consconfbounds} and Appendices~\ref{sec_app_C}, \ref{app_D}).
    

\end{enumerate}


The paper's outline: Section~\ref{sec_relwork} gives critical context, while Section~\ref{sec_review_CBI} reviews CBI. Section~\ref{sec_KlotzModel} introduces a statistical model for (possibly) dependent system failures/successes. The conservative confidence bounds on \emph{pfd} in Section~\ref{sec_consconfbounds} are applied in Section~\ref{sec_results}, and discussed in Section~\ref{sec_discussion}. Section~\ref{sec_conc} concludes the paper.

\section{Related Work}
\label{sec_relwork}
\subsection{The i.i.d. Assumption in Reliability Assessments}
Statistical models with the i.i.d. assumption have a long history of being used in software reliability assessment. Thayer \emph{et al.}'s~\cite{ThayerEtAl_1975} model was one of the earliest, used in early works on random testing \cite{DuranEtal_1984_EvaluationofRandomTesting}. Regulatory bodies have recommended using the i.i.d. assumption in reliability assessments when appropriate~\cite{atwood2003handbook}. 

But there are reasons to doubt the assumption. For instance, the possibility of ``failure clustering''; where a system receives sequences of inputs from its operational environment that cause the system to fail. These inputs form trajectories through a system's \emph{failure region} -- the subset of inputs that cause system failure. Failure regions can have topologically interesting properties that allow for failure clustering \cite{AmmanKnight_FailureClustering_1988,Bishop_failureclustering_1993}. These observations informed random testing approaches~\cite{Huang_Sun_2021_AdaptiveRandomTestingSurvey} and approaches for assessing systems with recovery block fault-tolerance~\cite{Csenki_RecoveryBlocks_1993,TomekTrivedi_recoveryblocks_1993}. 

\subsection{Weakening the i.i.d. Assumption: Statistical Models}
Various statistical models that weaken the i.i.d. assumption include: Chen and Mill's binary Markov chain~\cite{chen_binary_1996}, Goseva-Popstojanova and Trivedi's Markov renewal process~\cite{goseva_popstojanova_failure_2000} (extended in \cite{Xie_2005_ModellingCorrelatedFailures} by Xie \emph{et al}.), and Bondavalli \emph{et al}.'s Markov model~\cite{bondavalli_dependability_1995} with benign-failure states. Classical statistical inference produces ``\emph{point estimates}'' for these models' parameters, using only data from operational testing. Such estimates do not reflect an assessor's uncertainty about whether the i.i.d. assumption holds or not. Indeed, fitted parameter values imply that the models either exhibit dependence, or they do not -- there is no room for uncertainty here. In contrast, Bayesian inference -- our preferred approach -- allows for such uncertainty. In this paper, within a Bayesian framework, we model the system's failure process as a Markov chain introduced by Klotz~\cite{klotz_statistical_1973}. The Klotz model predates, is consistent with, and has (at most) as many states/parameters as, the models in~\cite{chen_binary_1996,goseva_popstojanova_failure_2000}.

Another advantage of Bayesian approaches is the assessor's beliefs (about the unknown \emph{pfd}) are explicitly accounted for in the assessment; beliefs that are justified by various forms of reliability evidence. And, although models of dependent system failures/successes have been developed, none of the assessment approaches using these models provide demonstrably \emph{conservative} statistical support for the skeptical assessor. Particularly when such support is justified by operational testing and other forms of reliability evidence. To the best of our knowledge, ours is the first approach to guarantee conservatism in the face of i.i.d. uncertainty.

\subsection{(Conservative) Bayesian Methods for Assessments}
Bayesian methods have been applied in various assessment scenarios, e.g. \cite{porn_two-stage_1996, bunea_two-stage_2005} involve hierarchical models, while \cite{WMiller_1992_BayesianReliabilityAssesssment,singh_2001_BayesianAssessment,LittlewoodPopov_2002_AssessingTheReliabilityDiversesoftware,Popov_2013_BayesianReliabilityAssessement} all use families of Beta prior distributions. More recently, \emph{conservative Bayesian inference} (CBI) methods have been developed to: {\bf i)} address the usual challenge of eliciting a suitable prior distribution from an assessor -- a prior that captures all, and only all, of an assessor's beliefs/views about the \emph{pfd}; {\bf ii)} give support for the most pessimistic reliability claims allowed by the reliability evidence. Thereby, CBI prevents dangerously optimistic claims.

Bishop \emph{et al.}~\cite{bishop_toward_2011} introduced CBI, illustrating its use in assessing safety-critical systems. Povyakalo \emph{et al.}~\cite{strigini_software_2013} use CBI to obtain the smallest probability of the system surviving future demands, and Salako~\cite{Salako_QEST_2020} applies CBI to assessing binary classifiers. While Flynn \emph{et al.}~\cite{zhao_assessing_2019,zhao_assessing_2020} apply CBI in the assessment of autonomous vehicle safety. Littlewood \emph{et al.}~\cite{littlewood_reliability_2020} show how CBI supports dependability claims when evidence suggests a new system is an ``improvement'' over an older system it replaces. Salako \emph{et al}.~\cite{salako_conservative_2021} extend this work, by considering more general ``improvement arguments'' for a wider range of assessment scenarios. Our application of CBI differs from these in 2 important ways: it allows for dependent testing outcomes during operational testing, and it incorporates skepticism of the i.i.d. assumption into assessments.
\kizito{TODO: It is obvious we are the authors of this paper! :-) Must add more references for the camera ready version}

CBI methods are closely related to \emph{robust Bayesian analysis}~\cite{berger1994_robustBayesianOverview,berger1987_SensitivityToPrior,Lavine_1991,Berger_1994_RobustnessInBidinesionalModels}, which studies how the results of Bayesian inference are impacted by uncertainty about the inputs to inference -- inputs such as the prior distribution and the statistical model. In particular, Lavine~\cite{Lavine_1991} outlines methods that reveal how uncertainty about the statistical model (specifically, about the so-called sampling distribution) impacts inference. This uncertainty is represented by a suitably general joint prior distribution over the family of stochastic processes defined by the statistical model. Subject to constraints on this prior, algorithms give the largest and smallest values for posterior measures of interest -- e.g. posterior expectations. Our use of CBI parallels the statistical techniques of Lavine, but applied to reliability assessments. Also, Lavine considers likelihoods consisting of products of the same functional form, while we do not (in order to weaken the i.i.d. assumption).  

Draper~\cite{Draper_1995} also tackles the problem of statistical model uncertainty in Bayesian inference, but offers an alternative solution. If one is uncertain about a model's assumptions -- specifically, assumptions that constrain the structural/functional forms of the related family of stochastic processes -- one could replace the model with an expanded model. This expanded model encompasses all of the stochastic processes defined under the original model, as well as other stochastic processes that violate the assumptions in question. A suitable prior distribution over this expanded model has to be defined. Draper argues for this Bayesian hierarchical model as a way of addressing model uncertainty. We concur, and further argue for the inference to be conservative; our results are guaranteed to be conservative, while Draper's are not.

For uncertainty about finitely many alternative models, Pericchi \emph{et al.}~\cite{Pericchi_1994} use discrete prior distributions over these alternatives. In principle, this is a hierarchical model akin to Draper's approach, but in more abstract terms within the robust Bayes framework. Our results lie at the intersection of Pericchi \emph{et al.} and Draper's ideas, within a reliability assessment context.

\subsection{On-demand vs Continuously Operating Software}
This work focuses on assessing on-demand systems: i.e. systems that do not operate continuously, taking action only when certain operating conditions (i.e. demands) arise~\cite{iectr63161_2022,2014_Rausand_relOfSafetyCriticalSys}. Reliability assessments for such systems can use ``discrete-time'' statistical models (e.g. \emph{Bernoulli processes}) with appropriate reliability measures (e.g. \emph{pfd}). Contrast this with continuously operating software\footnote{Such software \emph{can} have downtime due to, say, maintenance or upgrades.}~\cite{MichaelLyu_1996}; for assessing \emph{these} systems, it is more appropriate to employ ``continuous-time'' statistical models (e.g. \emph{non-homogeneous Poisson processes}) and consider reliability measures like failure-rates. \emph{Software reliability growth models} (SRGMs) are an extensive family of stochastic processes used in predicting the future reliability of continuously operating, evolving software (e.g. software with bugs that are discovered and fixed overtime). If bug fixing is successful \emph{without} introducing new software bugs then, \emph{ceteris paribus}, the software becomes more reliable with each fix; i.e. reliability ``grows''. Singpurwalla and Wilson~\cite{1994_Sinpurwalla_RelGrowthModels} give an overview of early SRGMs, while Miller~\cite{1986_Miller_ExpOrderStatistics} details a unified mathematical characterisation of large SRGM subfamilies. Also, see Bergman and Xie's review of early Bayesian SRGMs~\cite{1991_BergmanXie_BayesianSRGMs}.

\section{Review: Bayesian Reliability Assessment}
\label{sec_review_CBI}
Statistical inference for reliability claims comes in different flavours. The  classical ``frequentist'' confidence statement, e.g. 95\% confidence in a \textit{pfd} bound $b$, typically means that with a sufficiently large number of i.i.d. tests, there is no more than a 5\%  chance that the software succeeds on all the tests despite having a \emph{pfd} worse than $b$. While the Bayesian approach, instead, treats \textit{pfd} as a random variable, with a ``prior'' probability distribution representing an assessor's evidence-based beliefs about the \emph{pfd} before operational testing. The assessor updates their beliefs (via Bayes Theorem) upon seeing testing evidence. This yields a ``posterior'' distribution. Reliability claims can be made using this posterior; claims that reflect the assessor's updated judgements. For example, after seeing a sufficiently large number of successful tests, the assessor's Bayesian ``confidence'' in a \emph{pfd} bound $b$ -- i.e. their conditional probability of the \emph{pfd} being less than $b$ -- is 95\%.

Specifically, we recall the standard Bayesian approach to assessing an on-demand system~\cite{atwood2003handbook,littlewood_validation_1993,littlewood_conservative_1997}, in which the i.i.d. assumption is adopted. An i.i.d. Bernoulli process represents the stochastic failure behaviour of the system's software. We denote by $X$ the system's unknown \emph{pfd} due to software failures. According to the operational profile~\cite{musa_operational_1993}, $n$ demands are randomly submitted to the software and no failures are observed (this is the usual requirement when assessing a safety-critical system using operational acceptance testing). 
Let $b$ be the required upper bound on \emph{pfd}. 
Bayesian inference then gives an assessor's posterior confidence in $b$ after observing $n$ tests without failure\footnote{... as well as $P(\mbox{\emph{ failure-free operation}}\mid\,n\mbox{ \it demands without failure})$.}:
\begin{align}
\label{eq_post_cf_bound_with_complete_prior}
P (X &\leqslant b \mid \,n\mbox{ \it demands without failure}) \nonumber \\
&=\frac{P(X\leqslant b,\,n\mbox{ \it demands without failure})}{P(n\mbox{ \it demands without failure})}
\end{align}

In practice, Bayesian reliability assessments require that one specifies a prior distribution representing one's beliefs about the possible values of \textit{pfd}. CBI relaxes this 
 by requiring only a \textit{partial specification} of the prior distribution, when such specifications can be justified by evidence obtained prior to operational testing. Such partial specifications -- so-called  \emph{prior knowledge} (PK) -- take various forms. Most notably, the form of confidence bounds; e.g. being 90\% confident that the \emph{pfd} is no greater than $10^{-3}$, partly because IEC 61508 Safety Integrity Level 3~\cite{iec61508:2010} was a strict requirement in the system's development. When an assessor articulates their beliefs as PKs, there is an \textit{infinitely} large set $\mathcal D$ of all prior distributions that satisfy these PKs. CBI then determines the worst support the priors in $\mathcal D$ can give for a reliability claim; e.g. the smallest posterior confidence \eqref{eq_post_cf_bound_with_complete_prior} an assessor can have:
\begin{equation}
\label{eqn_GenCBIProblem}
\inf_{\mathcal D}P (X \leqslant b \mid \,n\mbox{ \it demands without failure})
\end{equation}
The solution of \eqref{eqn_GenCBIProblem} identifies a prior with posterior confidence that is the infimum value; \emph{no other prior that satisfies the PKs can give a smaller value for posterior confidence} \eqref{eq_post_cf_bound_with_complete_prior}. In this sense, CBI results are conservative. This ``worst case'' prior encodes within it the most conservative assessor beliefs consistent with the PKs. 




\section{A Statistical Model of Testing Outcomes}
\label{sec_KlotzModel}
The Klotz model \cite{klotz_statistical_1973} is a stationary random process consisting of possibly dependent Bernoulli trials. As a model of a system's failure process it generalises the i.i.d. Bernoulli process used in reliability assessments. During operational testing, a random sequence of demands is submitted to the system. On each demand, the system either successfully handles the demand or fails. So, we have a sequence of random variables $T_1, \ldots, T_n$, each taking the values 0 or 1, corresponding to success or failure, respectively. In this paper we follow the usual convention of upper-case letters for random variables and lower-case for their realisations. 

The Klotz model is characterised by a ``frequency'' parameter $x$ and a ``dependence'' parameter $\lambda$. Here, $x$ is the system's \emph{pfd} while $\lambda$ is the probability that a failure is followed by another failure. So, $P(T_1=1)=P(T_i=1)=x$ and $P(T_i=1\mid T_{i-1}=1)=\lambda$ for $i=2,\dots,n$. By requiring the process be \emph{1st-order stationary}, we have $P(T_i=1\mid T_{i-1}=0)=\frac{(1-\lambda)x}{1-x}$ for $i=2,\dots,n$ (see appendix \ref{app_KlotzModel}),
which yields Figure~\ref{fig_klotz}.
\begin{figure}[!h]
  \centering
  \includegraphics[width=0.3\textwidth]{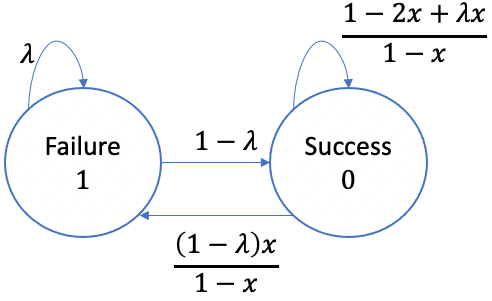}
  \caption{The Klotz model with dependent Bernoulli trials~\cite{klotz_statistical_1973}.}
  \label{fig_klotz}
\end{figure}
The transitions in Figure~\ref{fig_klotz} need to lie between zero and one, so
\begin{align}
\label{eq_ranges_x_lambda}
0\leqslant x <1, \quad \max \left\{0,(2x-1)/x\right\} \leqslant\lambda\leqslant1
\end{align}
Inequalities \eqref{eq_ranges_x_lambda} define a subset $\mathcal R$ of the unit square (Figure~\ref{fig_RregionAndPhicons}). 

\begin{figure*}[htbp]
\captionsetup[figure]{format=hang}	
    \centering
    \begin{subfigure}[h!]{0.33\linewidth}
	\centering	\includegraphics[width=0.75\linewidth]{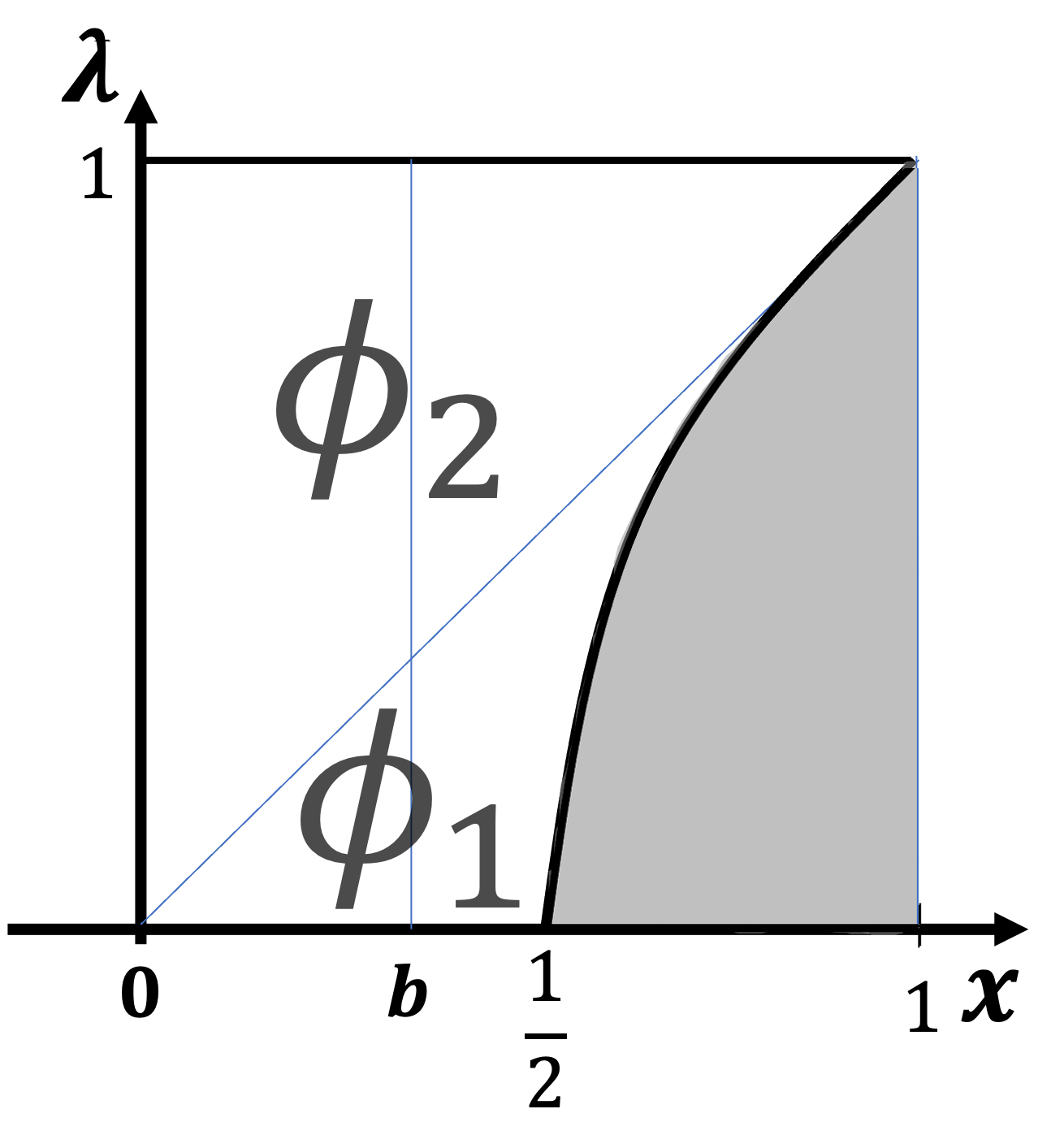}
	\caption{{\footnotesize  \normalsize}}	
	\label{fig_RregionAndPhicons}
	\end{subfigure}
	\begin{subfigure}[]{0.33\linewidth}
	\centering	\includegraphics[width=0.75\linewidth]{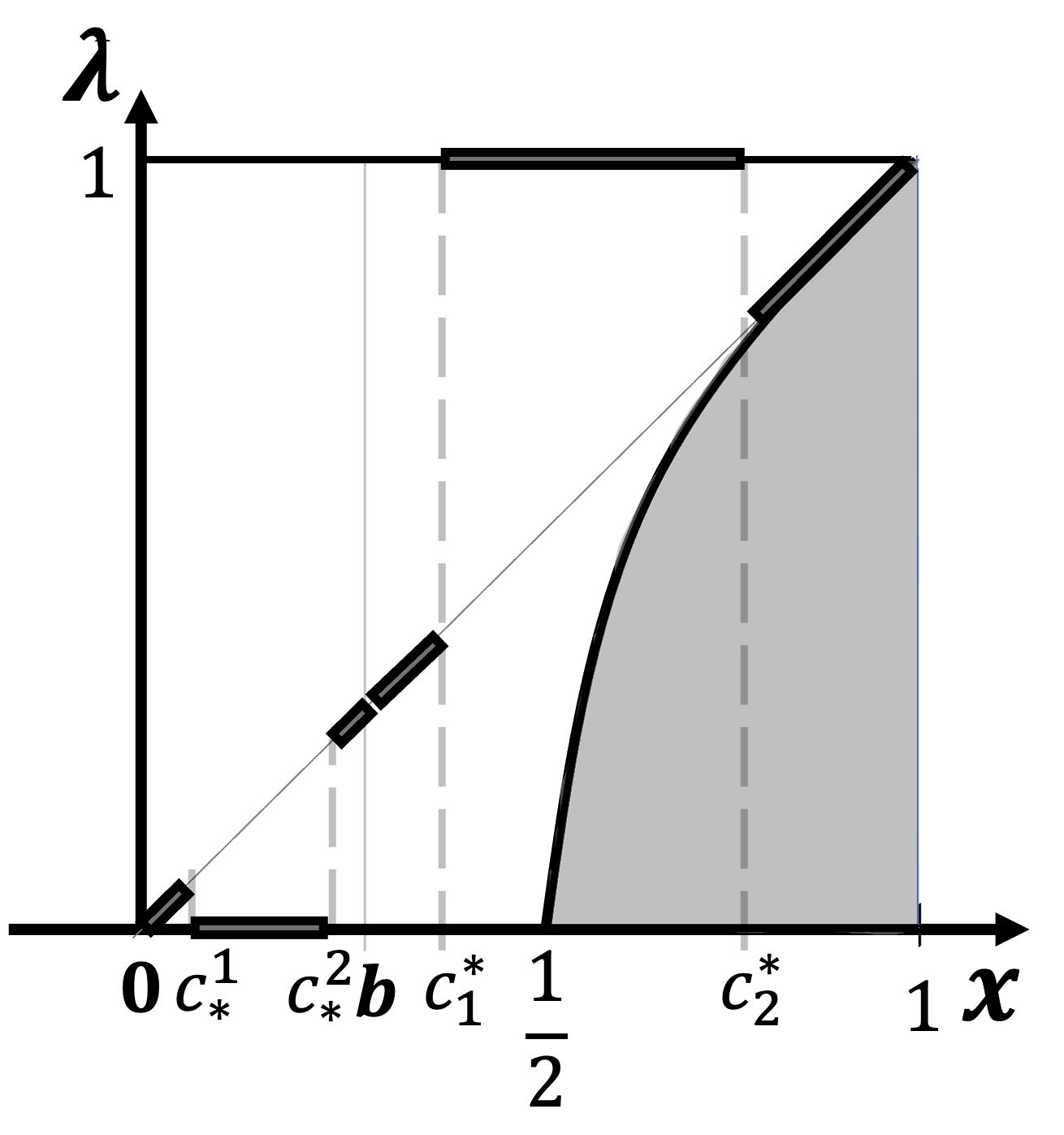}
	\caption{{\footnotesize  \normalsize}}
	\label{fig_continousmarginal}
	\end{subfigure}
 \begin{subfigure}[]{0.33\linewidth}
	\centering	\includegraphics[width=0.75\linewidth]{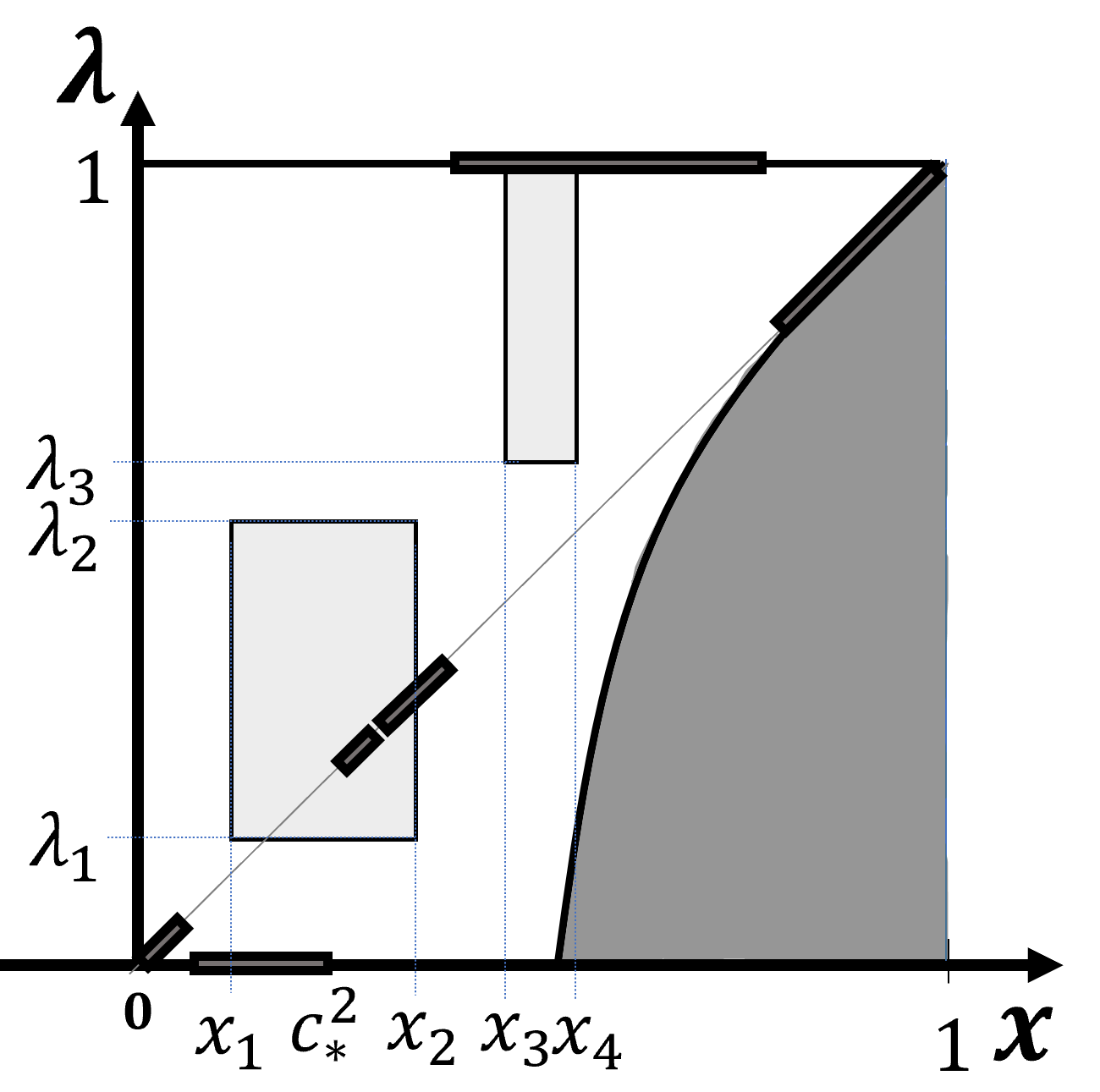}\vspace{0.25cm}
	\caption{{\footnotesize  \normalsize}}
 \label{fig_continousmarginal_example}
	\end{subfigure}
\caption[partition of {\mathcal R}]{{\footnotesize {\bf a)} A depiction of the region $\mathcal R$ defined by inequalities \eqref{eq_ranges_x_lambda}, and the subsets of $\mathcal R$ defined by PKs \ref{cons_negative_dependence} and \ref{cons_positive_dependence}. All prior distributions of $(X,\Lambda)$ have domain $\mathcal R$; {\bf b)} A prior distribution that gives the theorem's infimum -- depicted ``from above'', looking down on its domain $\mathcal R$. The \emph{pfd}s $c_{\ast}^1, \ c_{\ast}^2, \ c^{\ast}_1$, $c^{\ast}_2$ satisfy the theorem's constraints; {\bf c)} Probabilities of the shaded rectangular regions are given by integrals of $f(x)$ from PK\ref{cons_fully_specified_marginal}
.  \normalsize}}  
\label{fig_combined_continousmarginal_RregionAndPhicons}
\end{figure*}

Suppose the system succeeds on all $n$ demands during operational testing. For $(x,\lambda)\in{\mathcal R}$, the Klotz likelihood function gives the probability of observing this sequence of successes\footnote{We define $L(1,1;n):=\lim_{(x,\lambda)\rightarrow (1,1)}L(x,\lambda;n)=0$, where the limiting process involves only $(x,\lambda)$ values in $\mathcal R$.}: 
\begin{align}
L(x,\lambda ; n) {}:={} (1-x)\left(1-\frac{(1-\lambda)x}{1-x}\right)^{n-1}
\label{eqn_xKlotzlklhdFn_maintxt}
\end{align}
Different values of $(x,\lambda)$ can alter the dependence structure and functional form of the Klotz likelihood: {\bf i)} $x=\lambda$ is the special case of independent testing outcomes, with $L(x,x;n) {}={} (1-x)^n$; {\bf ii)} when $\lambda>x$, successes (and failures) tend to cluster during testing -- i.e. positive dependence. In the extreme, $\lambda=1$ and we have $L(x,1;n) {}={} (1-x)$; {\bf iii)} lastly, when $\lambda<x$, successes and failures tend to alternate more often -- i.e.  negative dependence. In the extreme, $x=\frac{1}{2-\lambda}$ and we have $L(\frac{1}{2-\lambda},\lambda;n)=0$.

The assessor's uncertainty about $x$ and $\lambda$ is captured by a joint prior distribution of $(X,\Lambda)$ over ${\mathcal R}$. The assessor's posterior confidence \eqref{eq_post_cf_bound_with_complete_prior} is:
\begin{align}
P(X\leqslant b\,\mid\,n\mbox{\it\ demands without failure}) 
=\frac{\MyExp[L(X,\Lambda;n){\bf 1}_{X\in[0,b]}]}{\MyExp[L(X,\Lambda;n)]}
\label{eqn_Klotzmodelpostconf}    
\end{align}
where the indicator function ${\bf 1}_{x\in\tt S}$ equals 1 if $x\in {\tt S}$, and is 0 otherwise. 

In Section \ref{sec_consconfbounds}, via CBI (i.e. constrained nonlinear optimisation of \eqref{eqn_Klotzmodelpostconf} in the vein of \eqref{eqn_GenCBIProblem}), we derive conservative posterior confidence bounds on \emph{pfd}. 

\section{Conservative Confidence (pfd) Bounds}
\label{sec_consconfbounds}
 We now present conservative posterior confidence bounds on \emph{pfd} for the skeptical assessor. We begin by formalising the assessor's beliefs as a collection of 4 constraints called ``prior knowledge''. These PKs only weakly specify the joint prior distribution of $(X,\Lambda)$. Firstly, the assessor's prior distribution of \emph{X} is continuous.\xingyu{not just beta...}

\begin{constraint}[continuous prior distribution of $X$]
\label{cons_fully_specified_marginal}
A density function $f(x)$ gives the assessor's prior confidence in pfd bound $u$ -- i.e. $P(X\leqslant u)=\int_{0}^{u}\!f(x)\diff x\ $ for all $u\in[0,1]$.
\end{constraint}
 
Secondly, the assessor does not rule out negatively or positively correlated testing outcomes (see Figure~\ref{fig_RregionAndPhicons}).
\begin{constraint}[confidence in negative correlation]
\label{cons_negative_dependence}
{$\phi_1 \times 100\%$}~confident that the outcomes of successive tests are negatively correlated, i.e. $P(\Lambda < X)=\phi_1$.
\end{constraint}
\begin{constraint}[confidence in positive correlation]
\label{cons_positive_dependence}
{$\phi_2 \times 100\%$}~confident that the outcomes of successive tests are positively correlated, i.e. $P(\Lambda > X)=\phi_2$.
\end{constraint}

Thirdly, the assessor is relatively confident that the i.i.d. assumption holds, and bound $b$ is satisfied, but not so confident as to make testing unnecessary. A straightforward extension of the theorem presented here accounts for less confidence.\xingyu{and relax this PK constraint...}
\kizito{TODO: wouldn't mind figuring out a better way to describe this PK}
\xingyu{Prior confidence in b is sufficient to do the testing whithout breaking things, and at the same time not with quite high confidence making testing unnecessary...}
\begin{constraint}[confidence in bound and independence]
\label{cons_reliable_system}
For a target system pfd $b$, $\phi_1 \leqslant P(X\leqslant b) \leqslant 1-\phi_2$
\end{constraint}
\xingyu{I wonder if we would like to present this restriction as a PK, as it may make our contribution looks a bit restricted? I was thinking: our problem statement is general enough (with only PK1 and PK2), although we only present the numerical results for the case with PK3 (because they are practical). While we also have the analytical result for the cases without PK3 (already in the appendix?), but to make the paper less overwhelming we choose to omit the numerical results. I feel a footnote/short-paragraph in the experiments section will do the job? Or I might be being naive or crafty...}\kizito{Yes, I see that PK3 could be criticised as being too restrictive. But the proof in the appendix uses PK3 -- specifically, PK3 means that the ``$f(x)$'' integral constraints in the ``argmin'' sub-optimisations are guaranteed to exist for some $r<s\leqslant b$ and $b\leqslant v<w$. So the results and proofs of the paper are really for the PK3 case. Like you suggest, let's think about how and where in the paper we make it clear that PK3 is not really that much of a restriction, e.g. pointing out how the cases where PK3 is violated are trivial extensions of the Theorem we present}

An assessor with these 4 beliefs has conservative posterior confidence bounds given by this theorem (see appendix \ref{sec_app_C}):

\begin{theorem*}
Let $\mathcal D$ be the set of all prior distributions over $\mathcal R$ and assume $0<b<\sfrac{1}{2}$. Using \eqref{eqn_Klotzmodelpostconf}, the optimisation problem
\begin{align*}
&\inf\limits_{\mathcal D} P(\,X\leqslant b \mid n\mbox{ demands without failure} ) \\
\mbox{s.t.} \,\,\,\,\,\,\,\,\,\,\,&PK\ref{cons_fully_specified_marginal},\,\,PK\ref{cons_negative_dependence},\,\,PK\ref{cons_positive_dependence},\,\,PK\ref{cons_reliable_system}
\end{align*}
is solved by the prior in Figure~\ref{fig_continousmarginal}, since the infimum equals the value of $P(\,X\leqslant b \mid n\mbox{ demands without failure} )$ for this prior. The infimum takes the form $\frac{1}{1 + Q}$, where $Q$ is

\begin{align}
\frac{\int_{b}^{1}\!\left((1-x)^n{\bf 1}_{x\in(b,c_1^{\ast})\cup(c_2^{\ast},1)}\ +\ (1-x) {\bf 1}_{x\in(c^{\ast}_1,c^{\ast}_2)}\right)f(x)\diff x}{\int_{0}^{b}\!\left((1-x)^n {\bf 1}_{x\in(0,c^1_{\ast})\cup(c^2_{\ast},b)}\ +\ \frac{(1-2x)^{n-1}}{(1-x)^{n-2}}{\bf 1}_{x\in(c^1_{\ast},c^2_{\ast})}\right)f(x)\diff x}
\label{eqn_CBIsoln_continousmarginal_1}
\end{align}
and the pfds $c^1_{\ast}, c^2_{\ast}, c_1^{\ast}, c_2^{\ast}$, are the unique values of $r,s,v,w$, respectively, that solve
\[\begin{array}{ccc}
&\argmin\limits_{0\leqslant r<s\leqslant b}\mid g_l(r)-g_l(s)\mid,\,\argmin\limits_{b\leqslant v<w\leqslant 1}\mid g_u(v)-g_u(w)\mid&\\
\mbox{s.t.} \,\,\,&g_l(0)\leqslant g_l(r),\,\,\,\,g_l(b)\leqslant g_l(s)\,,&\\
&g_u(b)\leqslant g_u(v),\,\,\,\,g_u(1)\leqslant g_u(w)\,,&\\ 
&\int_{r}^{s}\!f(x)\diff x = \phi_1\,,\,\,\,\,\int_{v}^{w}\!f(x)\diff x = \phi_2\,,& \\
&0\leqslant r < s \leqslant b\leqslant  v < w \leqslant 1&
\end{array}\]
for $g_l:[0,\sfrac{1}{2}]\to [0,1]$ and $g_u:[0,1]\to [0,1]$ defined as
\begin{align}
g_l(x) &= (L(x,x;n)-L(x,0;n)){\bf 1}_{x\in[0,\frac{1}{2}]}\nonumber\\
g_u(x) &= (L(x,1;n)-L(x,x;n)){\bf 1}_{x\in[0,1]}
\label{eqn_thegfunctions}
\end{align}
\end{theorem*}
Numerical estimates for $c^1_{\ast}$, $c^2_{\ast}$, $c_1^{\ast}$, $c_2^{\ast}$ may be computed using ``root-finding'' algorithms such as that in Appendix \ref{app_D}. 

\begin{figure*}[h!]
     \centering
     \begin{subfigure}[b]{0.31\textwidth}
         \centering
         \includegraphics[width=\textwidth]{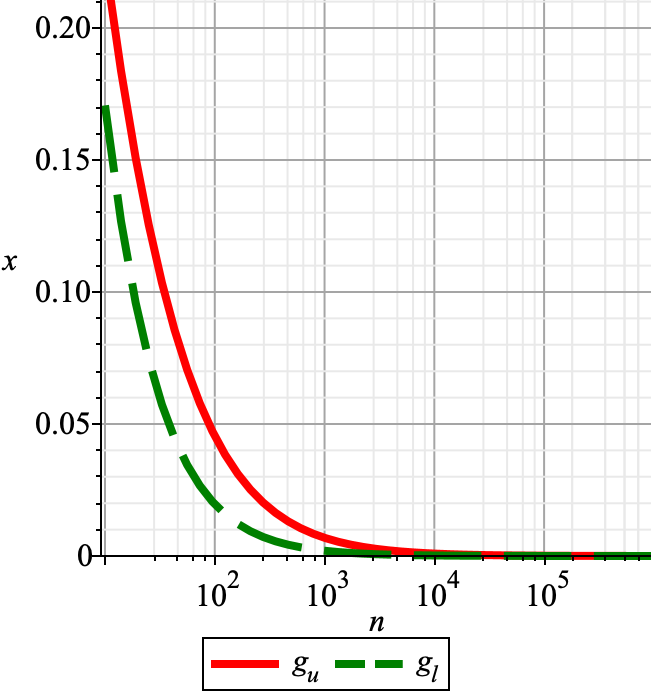}
         \caption{stationary points for $g_l$ and $g_u$}
         \label{fig_glocalmaximaxvalues}
     \end{subfigure}
     \hfill
     \begin{subfigure}[b]{0.31\textwidth}
         \centering
         \includegraphics[width=\textwidth]{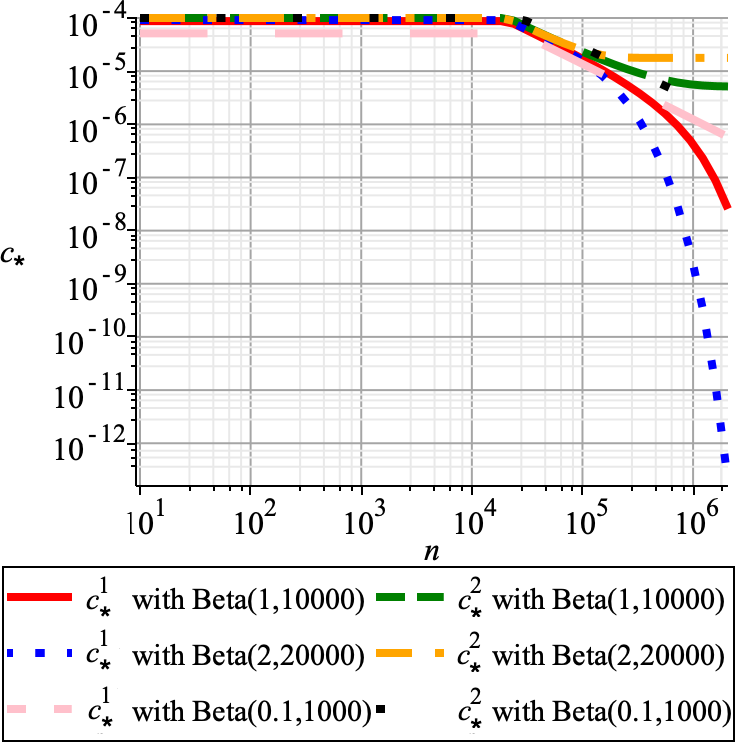}
         \caption{$\phi_1=0.05$, $\phi_2=0.05$, $b=0.0001$}
         \label{fig_cstarlowerplots}
     \end{subfigure}
     \hfill
     \begin{subfigure}[b]{0.31\textwidth}
         \centering
         \includegraphics[width=\textwidth]{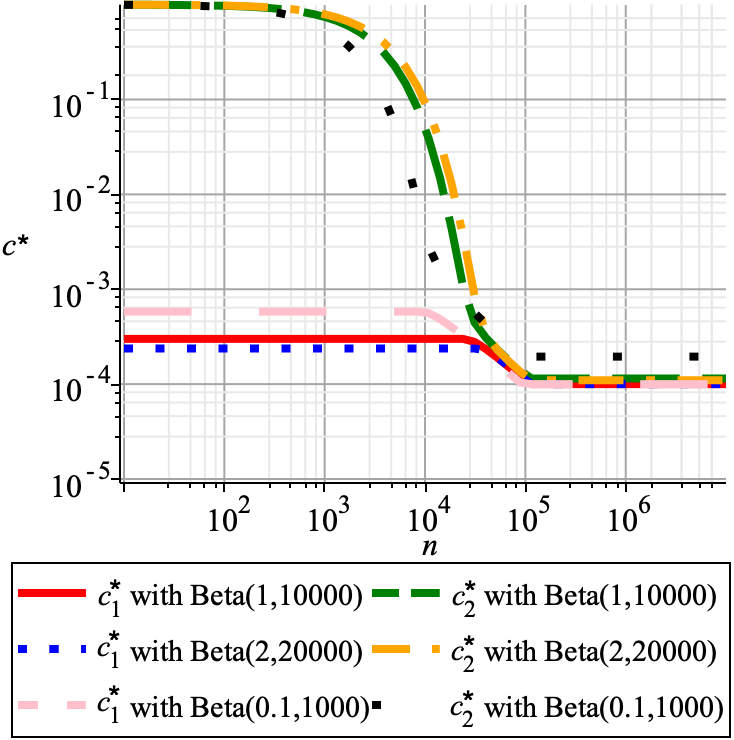}
         \caption{$\phi_1=0.05$, $\phi_2=0.05$, $b=0.0001$}
         \label{fig_cstarupperplots}
     \end{subfigure}
        \caption{The stationary points of $g_l$, $g_u$ and the \emph{pfd}s $c^1_{\ast}$, $c^2_{\ast}$, $c_1^{\ast}$, $c_2^{\ast}$ are all monotonically decreasing functions of $n$.}
        \label{fig_sensitivity_analysis_row1}
\end{figure*}

If the assessor has no doubts about the i.i.d. assumption, this CBI result is given by traditional Bayesian inference. That is,
\begin{corollary*}
Let $Q$ be given by \eqref{eqn_CBIsoln_continousmarginal_1} in the theorem, when $\phi_1=\phi_2=0$. Then the infimum is given by traditional Bayesian inference, using the prior density $f(x)$. That is,
\begin{align}
\label{eqn_CBIbecomesBayesianInf}
\inf\limits_{\mathcal D} P(X\!\leqslant\! b \mid n\mbox{ demands without failure} ) 
\!=\! \frac{\int_{0}^b (1-x)^n f(x)\diff x}{\int_{0}^1 (1-x)^n f(x)\diff x}
\end{align}
\end{corollary*} 
However, if the assessor begins operational testing with beliefs captured by PKs \ref{cons_fully_specified_marginal}, \ref{cons_negative_dependence}, \ref{cons_positive_dependence} and \ref{cons_reliable_system}    -- i.e. beliefs that only partially specify a joint prior distribution of $(X,\Lambda)$ -- then the CBI theorem identifies a joint prior that is consistent with these beliefs, is unique up to zero probability events, and that gives the smallest posterior confidence in the \emph{pfd} upper-bound $b$ (see Figure~\ref{fig_continousmarginal}).

This CBI prior assigns probability density only to the indicated thick black line segments in $\mathcal R$; the rest of $\mathcal R$ has zero probability\footnote{Expressing this formally involves projections from $\mathcal R$ to the interval $[0,1]$.} (Appendix \ref{sec_app_C}). For example, prior to testing, the conservative assessor expects negatively, or positively, dependent test outcomes with probabilities $\phi_1=P(c^1_\ast\leqslant X\leqslant c^2_\ast,\,\Lambda=0)=\int^{c^2_\ast}_{c^1_\ast}\!f(x)\diff x$ and $\phi_2=P(c_1^\ast\leqslant X\leqslant c_2^\ast,\,\Lambda=1)=\int^{c_2^\ast}_{c_1^\ast}\!f(x)\diff x$, respectively. Or consider the probabilities of the 2 shaded rectangular regions in Figure~\ref{fig_continousmarginal_example}. The non-zero contributions to these probabilities come from the thick black line segments that intersect these regions; i.e. $P(x_1\leqslant X\leqslant x_2,\, \lambda_1\leqslant \Lambda\leqslant \lambda_2)=\int_{c_{\ast}^2}^{x_2} \!f(x)\diff(x)$ and $P(x_3\leqslant X\leqslant x_4, \,\lambda_3\leqslant \Lambda\leqslant \,1)=\int_{x_3}^{x_4} \!f(x)\diff(x)$.
\kizito{to think of practical examples..; $c\_s$ are functions of numbers of failure-free tests and amount of phi doubts; If I see negatively correlated testing outcomes, the system pfd is likely to lie between $C^1_{*}$ and $C^2_{*}$ }
\xingyu{Would it be useful to write down the worse-case prior PDF as a piece-wise function:
\[   f_{joint}(x,\lambda)=\left\{
\begin{array}{ll}
      \int^{c^1_\ast}_{0}f(x)\diff x & x\leq c^1_{*}, \lambda=x \\
     \phi_1 & c^1_{*}\leq x\leq c^2_{*}, \lambda=0 \\
      \int^{b}_{c^2_\ast}f(x)\diff x & c^2_{*}\leq x\leq b, , \lambda=x \\
      \int^{c_1^\ast}_{b}f(x)\diff x & b\leq x\leq c^{*}_{1}, \lambda=x \\
      \phi_2 & c^{*}_{1}\leq x\leq c^{*}_{2}, \lambda=1 \\
      \int^{1}_{c_2^\ast}f(x)\diff x & c^{*}_{2}\leq x,  \lambda=x \\
      0 & otherwise \\
\end{array} 
\right. \]
Something like this, hoping the reviewers may know what the worst-case prior is even if he cannot follow the ``pop-up'' figure 2b..
}


\subsection{Conservative Beliefs for Failure-free Testing}
 The CBI prior Figure~\ref{fig_continousmarginal} encodes conservative beliefs. In particular, the \emph{pfd} values $c^1_{\ast}$, $c^2_{\ast}$, $c_1^{\ast}$, $c_2^{\ast}$ indicate where (in $\mathcal R$) doubts in the i.i.d. assumption should be placed for conservative posterior confidence. The $\mathcal R$ locations between $(c_\ast^1,0)$ and $(c_\ast^2,0)$ along the $\lambda=0$ edge represent statistical dependence that is \emph{unlikely} to produce failure-free testing \emph{if} the unknown \emph{pfd} actually satisfies the bound. And if the \emph{pfd} does not satisfy the bound, the locations between $(c^\ast_1,1)$ and $(c^\ast_2,1)$ along $\lambda=1$ represent dependence \emph{likely} to produce failure-free testing. Note that the assessor has been unable to rule out these extreme beliefs prior to testing, since these beliefs are consistent with the PKs. 

These beliefs depend on the number $n$ of required failure-free tests; the ``$c$''s become smaller if $n$ becomes bigger. Because, the ``$c$''s are defined with respect to the stationary points of the theorem's ``$g$'' functions -- themselves dependent on $n$. Figure~\ref{fig_glocalmaximaxvalues} plots how the $x$ values for these stationary points tend to zero as $n$ increases. Figure~\ref{fig_cstarlowerplots} shows the consequence of this -- $c_\ast^1$ decreases towards zero and $c_\ast^2$ decreases towards a non-zero value dependent on $f(x)$. Figure~\ref{fig_cstarupperplots} tells a similar story --  $c^\ast_1$ decreases towards the bound $10^{-4}$, and $c^\ast_2$ towards a non-zero value dependent on $f(x)$. \kizito{We need a better explanation here. The question is why do the cs reduce, in practical terms?}This asymptotic behaviour is reasonable: the greater the amount of failure-free operation required during testing, the smaller the \emph{pfd} is expected to be for a system that performs this well (even when one is doubtful of the tests being i.i.d.).


Informally, beliefs in i.i.d. tests ``lie on the boundary'' of the set of conservative beliefs that express doubts about i.i.d. tests, and confidence from beliefs in i.i.d. tests \eqref{eqn_CBIbecomesBayesianInf} is well-approximated by conservative confidence from the theorem. Indeed, by the \emph{dominated convergence theorem} \cite{schilling_2005},  \eqref{eqn_CBIsoln_continousmarginal_1} tends to $\frac{\int_{b}^1 (1-x)^n f(x)\diff x}{\int_{0}^b (1-x)^n f(x)\diff x}$ as $\phi_1$ and $\phi_2$ tend to 0 (since $c^1_\ast$, $c_1^\ast$ tend to $c^2_\ast$, $c_2^\ast$ respectively).


\section{Results: Assessment Using the Theorem}
\label{sec_results}
\xingyu{
\begin{itemize}
    \item As per we did in the DSN21 paper, we need blocks like Example 1, Example 2... In each example block, we should ask explicit/practical questions to show our use cases. For instance, I am brain storming:
    \begin{itemize}
        \item Example 1: Give those PKs, how much confidence we can have in a required pfd bound after seeing $n$ failure-free tests...figure 3 a and b.
        \item Example 2: Similar as Example 1, but how sensitive to $\phi_1$?, figure 3 c and d. \kizito{what's the impact of beliefs in negatively dependent testing outcomes on posterior confidence... to polish later.. }
        \item Example 3: Similar as Example 1, but how sensitive to $\phi_2$?, figure 3 e and f. \kizito{what's the impact of beliefs in positively dependent testing outcomes on posterior confidence... to polish later.. }
        \item Example 4: Regarding figure 4... maybe an example on the practical meaning behind c1star, etc... 
        \item Example 5: Take the Klotz model and ask how many tests we need to have 95\% confidence in 0.0001 in classical method... without new figures, rather discuss it as NWTES paper...
    \end{itemize}
\end{itemize}
}
\subsection{Practical Context and Guidance}




The theorem gives conservative confidence in a \emph{pfd} bound $b$, when an on-demand system is subjected to black-box operational testing. A bound such as $b=10^{-4}$; a target \emph{pfd} used in the assessment of the Sizewell-B nuclear power plant safety protection system in the United Kingdom\footnote{This was the target \emph{pfd} for a hardwired secondary safety subsystem; a software-based primary safety subsystem had a more modest $10^{-3}$ target \emph{pfd}.}\cite{Hunns_1992,Littlewood1998TheUseOfComps}. To gain $99\%$ confidence in this bound -- using the i.i.d. assumption under a classical statistical inference approach -- requires between $10^{4}$ and $10^{5}$ test demands \cite{parnas_1990,Littlewood1998TheUseOfComps}. In this section, we will use similar orders of magnitude of test demands to illustrate the theorem's use. 

In particular, the theorem can be used during acceptance testing to check the robustness of confidence \eqref{eqn_CBIbecomesBayesianInf} to doubts about the i.i.d. assumption; we illustrate how to do this in the rest of this section. When applying statistical techniques like the theorem, one may follow the guidance from Littlewood and Strigini~\cite{strigini_guidelines_1997} and Lyu \emph{et al}.~\cite{MichaelLyu_1996} on performing statistical testing. See Parnas \emph{et al}.~\cite{parnas_1990} for additional discussion on evaluating safety-critical software, including random test-case (e.g. demand) selection. For nuclear safety applications in particular, \cite{LicensingOfICSoftware,atwood2003handbook} give guidance on reliability assessment using statistical techniques. 

More generally, the theorem can be applied in any black-box testing phase where failure-free operation can be used to gain confidence (i.e. \eqref{eqn_CBIbecomesBayesianInf}) in the software. Such testing would typically involve subjecting (some part of) the software to a large number of randomly generated demands in a simulated environment. For example, in accordance with integrity level 4 (see IEEE 1012:2016~\cite{IEEEstd10122016}), one may apply the theorem during component, or integration, testing phases for safety-critical software. For nuclear safety-critical software testing phases, see also IEC 60880~\cite{IECstd60880}.

\begin{figure*}[h!]
     \centering
     \begin{subfigure}[b]{0.31\textwidth}
         \centering
         \includegraphics[width=\textwidth]{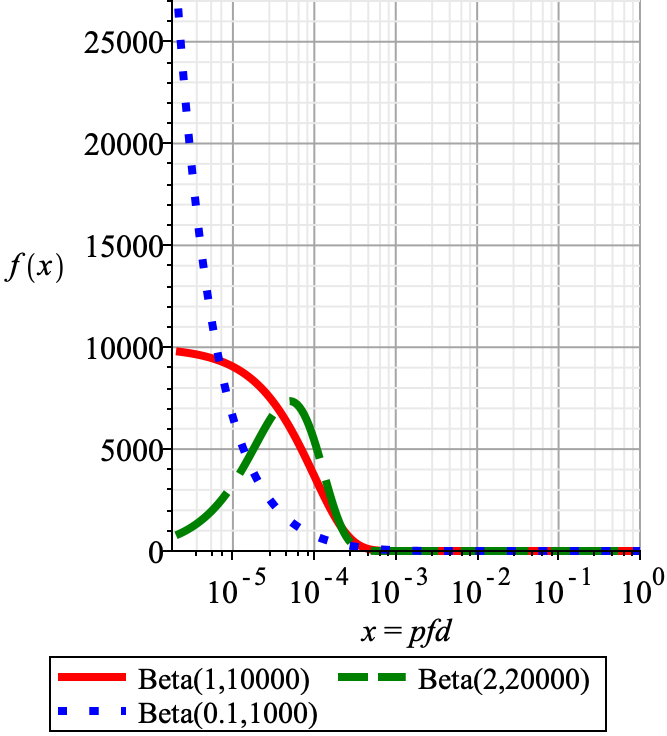}
         \caption{Beta prior distributions of \textit{pfd}}
         \label{fig_BetaPriors}
     \end{subfigure}
     \hfill
     \begin{subfigure}[b]{0.31\textwidth}
         \centering
         \includegraphics[width=\textwidth]{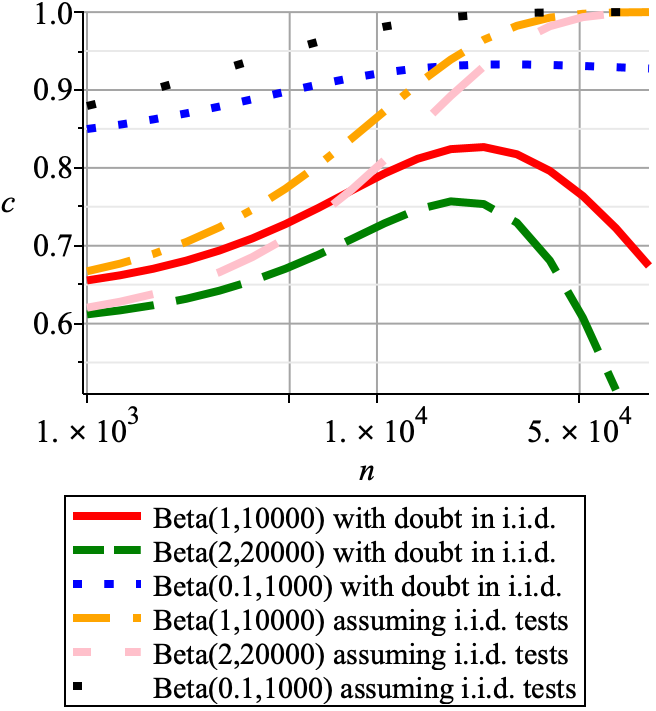}
         \caption{$\phi_1=0.05$, $\phi_2=0.05$, $b=0.0001$}
         \label{fig_vary_beta}
     \end{subfigure}
     \hfill
    \begin{subfigure}[b]{0.31\textwidth}
         \centering
         \includegraphics[width=\textwidth]{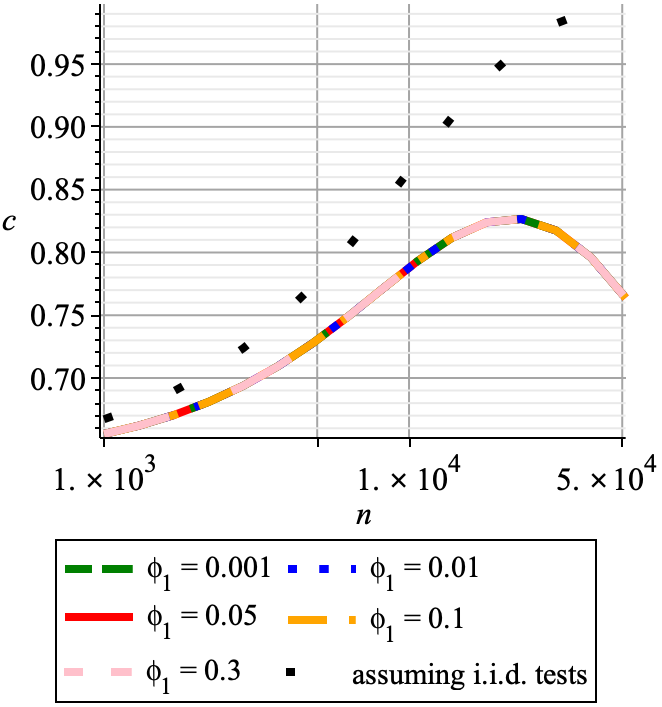}
         \caption{$\phi_2=0.05$, $b=0.0001$,  $\mathrm{Beta}(1,10000)$}
         \label{fig_vary_phi1_monotonic_beta}
     \end{subfigure}
     \hfill
     \par\bigskip
     \begin{subfigure}[b]{0.31\textwidth}
         \centering
         \includegraphics[width=\textwidth]{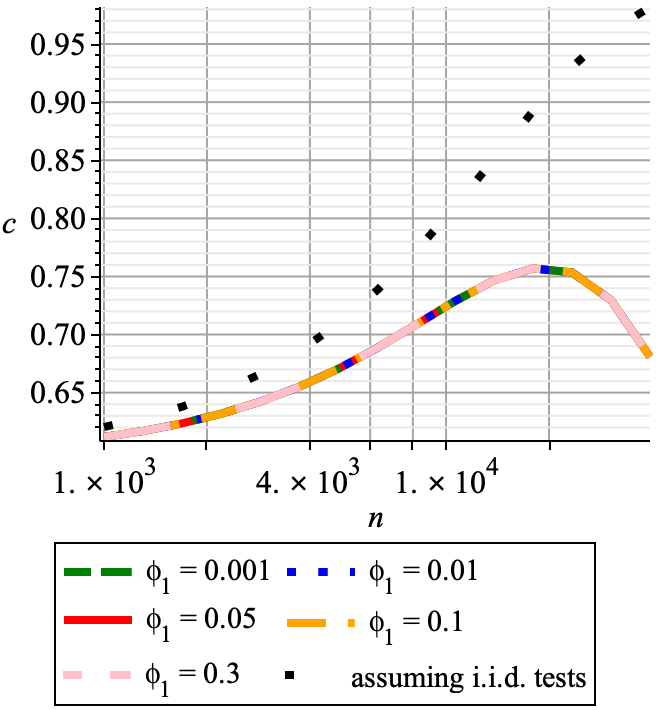}
         \caption{$\phi_2=0.05$, $b=0.0001$,  $\mathrm{Beta}(2,20000)$}
         \label{fig_vary_phi1_unimodal_beta}
     \end{subfigure}
     \hfill
     \begin{subfigure}[b]{0.31\textwidth}
         \centering
         \includegraphics[width=\textwidth]{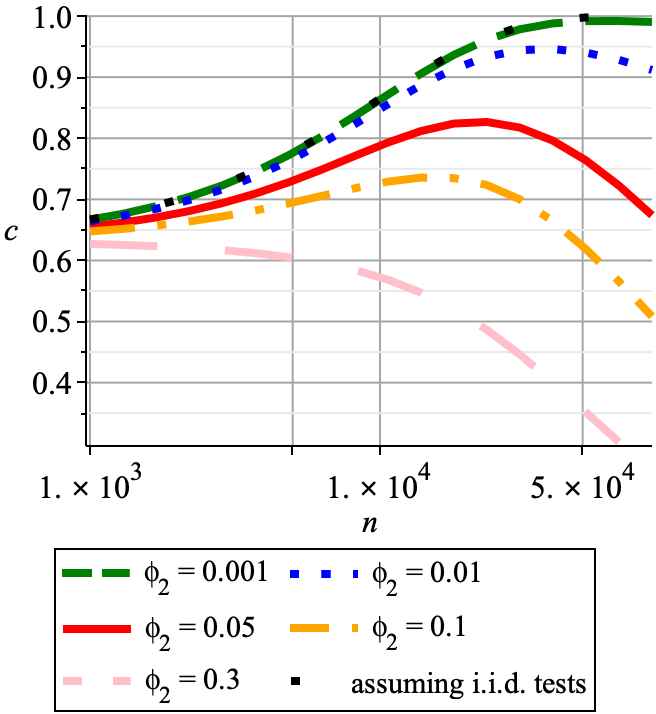}
         \caption{$\phi_1=0.05$, $b=0.0001$,  $\mathrm{Beta}(1,10000)$}
         \label{fig_vary_phi2_monotonic_beta}
     \end{subfigure}
     \hfill
     \begin{subfigure}[b]{0.31\textwidth}
         \centering
         \includegraphics[width=\textwidth]{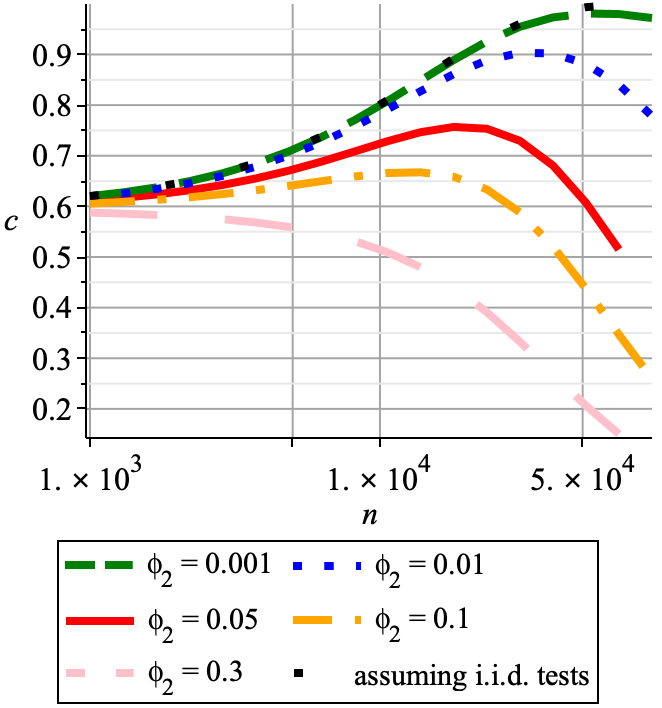}
         \caption{$\phi_1=0.05$, $b=0.0001$,  $\mathrm{Beta}(2,20000)$}
         \label{fig_vary_phi2_unimodal_beta}
     \end{subfigure}
        \caption{Sensitivity analyses showing which forms of PK have the biggest impact on (conservative) posterior confidence~$c$.}
        \label{fig_sensitivity_analysis_row3}
\end{figure*}





\subsection{Examples: Prior Beliefs and Confidence in a Bound}
The analyses in the rest of this section show how: {\bf i)} confidence based on the i.i.d. assumption can be very optimistic as failure-free tests accumulate; {\bf ii)} some forms of doubt about the i.i.d. assumption significantly impact confidence, while other forms do not; {\bf iii)} surprisingly, failure-free testing can eventually undermine confidence in the system satisfying the bound.

\begin{table}[h]
		\centering
		\caption{ A summary of 3 Beta prior distributions of \emph{pfd} \normalsize}
		\label{table_betapriors}
		\begin{tabularx}{0.33\textwidth}{cccc}
				\toprule[2pt]
				$\alpha$ & $\beta$ & $\MyExp[X]=\frac{\alpha}{\alpha+\beta}$ & $P(X\leqslant 10^{-4})$ \\ 
				\midrule
				2& 20000 & 0.0001	& 0.6\\
				1& 10000 & 0.0001	&	0.63\\
				0.1& 1000 & 0.0001 	& 0.83\\
				  		
				\bottomrule[2pt]
		\end{tabularx}
	\end{table}

The PK\ref{cons_fully_specified_marginal} density $f(x)$ can be from \emph{any family} of continuous distributions over the interval $[0,1]$. Beta densities are often used in practice \cite{atwood2003handbook,WMiller_1992_BayesianReliabilityAssesssment,LittlewoodPopov_2002_AssessingTheReliabilityDiversesoftware}. Consider 3 alternative $f(x)$; the Beta distributions in Figure~\ref{fig_BetaPriors} with parameters/properties in Table.~\ref{table_betapriors}. Let these represent prior beliefs of 3 assessors that differ in how confident they are that the system satisfies the bound. 

Suppose the assessors are a little skeptical of the tests being i.i.d. (e.g. $\phi_1=\phi_2=0.05$).  Figure~\ref{fig_vary_beta} shows how posterior confidence~$c$ evolves as operational testing evidence mounts. For each Beta prior, the posterior confidence \eqref{eqn_CBIbecomesBayesianInf} under an i.i.d. assumption is plotted against the CBI posterior confidence $\frac{1}{1+Q}$ for our skeptic. In all cases, confidence from assuming independence is initially comparable to conservative confidence -- the relevant pair of curves for each Beta prior almost overlap initially. 

However, as more failure-free testing is observed -- i.e. as $n$ grows -- ``i.i.d.''-based confidence grows and tends towards certainty. While conservative confidence grows more slowly, reaches a maximum, and then tends towards zero. Appendix \ref{sec_app_E} proves this zero limiting behaviour will occur whenever an assessor allows for the possibility of positively correlated tests (i.e. $\phi_2>0$). In the limit of large $n$, ``i.i.d.''-based confidence can be the \emph{most optimistic} confidence \emph{can} be, while still remaining consistent with an assessor's informed views/beliefs about the unknown \emph{pfd}.     

If the evidence available to an assessor before testing justifies being very confident in the bound, then the initial closeness between the ``i.i.d.''-based posterior confidence and the confidence given by CBI can continue for longer before ultimately diverging. In Figure~\ref{fig_vary_beta} the pair of curves for the Beta(0.1,1000) prior -- i.e. for the very confident assessor -- stay closer together for longer, compared to the curves for the least confident assessor with prior Beta(2,20000). The greater the prior confidence in the bound, the greater the posterior confidence when testing begins (i.e. small $n$).

\subsection{Sensitivity Analysis: How Failure-free Tests and Skepticism about i.i.d. Tests Impact Confidence}
The nature of an assessor's skepticism determines whether their conservative posterior confidence ultimately grows or shrinks during operational testing. At one extreme, some forms of doubt have no noticeable impact on confidence when no failures are observed during testing. The possibility (however likely) of negatively correlated tests has no apparent effect on conservative confidence. For example, an assessor might intentionally seek to ``stress'' the software during testing, by randomly including a disproportionate number of test demands that are thought will likely cause the software's failure. If stressful demands are adequately interspersed with significantly less stressful ones, one might expect the testing outcomes (i.e. the software's successes/failures) to exhibit some negative dependence (i.e. non-zero $\phi_1$) -- so a failure is quickly followed by successes, then another failure relatively soon afterwards, etc. However, when no failures occur, Figures~\ref{fig_vary_phi1_monotonic_beta},  \ref{fig_vary_phi1_unimodal_beta} show that the ``rise and fall'' of CBI confidence in Figure~\ref{fig_vary_beta} is unaffected by varying one's prior confidence in negative dependence. Intuitively, the more successes occur, the less likely these are from a system undergoing negatively dependent tests.

At the other extreme are positively correlated tests. Figures~\ref{fig_vary_phi2_monotonic_beta}, \ref{fig_vary_phi2_unimodal_beta} both show that the smaller $\phi_2$ is, the closer conservative confidence gets to the confidence under the i.i.d. assumption. When $\phi_2=0$, conservative confidence grows to certainty as the number of successes grows (Appendix~\ref{sec_app_E}). While the larger $\phi_2$ is, the more conservative confidence becomes. Here, confidence in positive correlations (i.e. large $\phi_2$) may be due to pessimistic reasons for the failure-free tests -- i.e. ``success clustering'' can occur even if the software is unreliable. The tests could be unrepresentatively ``easy'' for the software to correctly respond to, or the test oracle is incorrect so failures go undetected \cite{littlewood_use_2007},\cite{Barr_Oracleproblem_2015}.

So, failure-free testing can undermine one's confidence in a system's \emph{pfd}. Such conservatism is not unique to CBI -- even with classical inference, confidence bounds can be quite conservative initially, becoming optimistic (compared to the CBI bounds) after many tests. Indeed, using the Klotz likelihood \eqref{eqn_xKlotzlklhdFn_maintxt} when $\lambda=1$, the probability of succeeding on all $n$ tests (despite the \emph{pfd} being worse than $10^{-4}$) is at most $(1-10^{-4})=0.9999$. That is, a system with a \emph{pfd} worse than $10^{-4}$ may be almost certain to succeed on all tests \emph{if} the tests are strongly positively correlated. This bleak result holds \emph{for all} $n$; so, increasing the number of failure-free tests does not increase one's confidence in the system.

\section{Discussion}
\label{sec_discussion}
\subsection{Skepticism about Model Assumptions}
Software reliability assessments should be conservative: to wit, only when test results stand up to the most critical scrutiny can confidence be justifiably placed in the system satisfying a \emph{pfd} bound. Conservative assessments require a skeptical assessor. In Bayesian terms, our assessor holds conservative beliefs about what the evidence implies for a system's reliability, and about the validity of statistical modelling assumptions. 

This paper illustrates a general, incremental approach to dealing with doubts about \emph{any} statistical model assumptions: we offer a demonstrably conservative form of Draper's ideas \cite{Draper_1995}. For a model property one is doubtful of, one can check the sensitivity of claims based on the model by using a slightly more general model (that has the original model as a special case and weakens the property in question) for inference. ``Slight'' model generalisations keep models from becoming unnecessarily complex, ensure generalisations cover all scenarios covered by the original models, and minimize eliciting increasingly complex prior distributions.

This incremental approach is a ``win-win''. If the i.i.d. assumption is \emph{not} too optimistic, ``i.i.d.''-based confidence doesn't depart significantly from the CBI theorem's conservative confidence based on the Klotz model. If the i.i.d. assumption \emph{is} too optimistic, sensitivity analysis using the theorem can reveal this -- in such circumstances, caution is warranted when relying on ``i.i.d.''-based reliability claims. And if, in turn, one has doubts about the Klotz model, then a generalisation of the Klotz model can be used to check if ``Klotz model''-based confidence is sensitive  to doubts. 

\subsection{Limitations and Future Work}
Let us highlight some Klotz model shortcomings. It does not distinguish between different success/failure types. Future work might consider using the models of \cite{bondavalli_dependability_1995,bondavalli_modelling_1997,strigini_testing_1996} with CBI, to check the robustness of ``Klotz model''-based claims. These models account for the cumulative impact of benign failures. 


The Klotz model uses the relative sizes of $x$ and $\lambda$ to characterise \emph{all} pairs of successive Bernoulli trials as being identically positively, negatively, or zero correlated. 
Consequently, the Klotz model is unable to express non-stationary dependence, such as may be due to software updates that remove, or inadvertently add, faults to the software. The model cannot capture dependence across time either; such as periodic correlations over relatively short, or relatively long, runs of demands. Such periodicity can arise if demands that cause software failure are more likely at certain times when (or certain locations where) operating conditions tend to be more stressful (e.g. for software ensuring aircraft safety, unfavourable weather along a flight path may be more likely at certain times of the year, or more likely along certain flight paths). CBI models with time-dependent correlations are worth exploring.

On the use of \emph{pfd}s we make the following comment. When the failure process is stationary, \emph{pfd}s make sense. The probability of the system failing on the $n$-th demand is the same for all $n$. But for time-dependent failure probabilities, there are more suitable dependability measures -- such as the probability of future failure-free operation. Strigini \cite{strigini_testing_1996} makes a related point in a classical inference context. Even in the present context of posterior confidence bounds on \emph{pfd}, it's worth investigating whether other dependability measures are more/less robust to i.i.d. doubts.  

\xingyu{We might want to say more about this.}
We have formalised (via PK\ref{cons_fully_specified_marginal}) and illustrated how \emph{any} Bayesian reliability assessment using a continuous prior, $f(x)$, can conservatively incorporate doubts about ``i.i.d.''. However, when eliciting $f(x)$ proves too challenging, future work could extend the theorem to work with partially specified $f(x)$. For example, prior to testing, one might justifiably have some confidence in the software containing no faults \cite{littlewood_reasoning_2012}. In effect, $f(x)$ becomes discontinuous, with a non-zero probability of the \emph{pfd} being zero. For such scenarios, preliminary results suggest significantly better agreement between ``i.i.d.''-based and ``Klotz model''-based posterior confidence, even with a very large number of successful tests.

In general, the ``elicitation challenge'' remains an open problem for Bayesian approaches. In light of this, best practice approaches should be followed when eliciting PKs \cite{1994_NuRegTechreport,OHagan_2006}, while sensitivity analyses (as illustrated in Section~\ref{sec_results}) is crucial and practical for checking the robustness of confidence to PK changes.


This work has not considered model selection or validation. One might envisage applying CBI to conservatively gain confidence in the i.i.d. assumption, or using Bayes factors and CBI to conservatively determine which modelling assumptions lead to more trustworthy predictions about future system reliability.

The theorem could be extended to account for failures during testing, e.g. when assessing machine learning applications, or extended to support conservative claims for software modules in a fault-tolerant configuration (e.g. extending Singh \emph{et al}. \cite{singh_2001_BayesianAssessment}).




\section{Conclusions}
\label{sec_conc}
When assessing software using operational testing it is natural to ask, ``is it appropriate to assume the software's failures and successes arise in an i.i.d. manner?''. For many practical scenarios there are well-known reasons to doubt the i.i.d. assumption. A few statistical models which weaken the assumption have been proposed for use in reliability assessments~\cite{chen_binary_1996,goseva_popstojanova_failure_2000,bondavalli_dependability_1995,bondavalli_modelling_1997}. However, none of these proposals allow an assessor to remain unsure about whether i.i.d. holds or not, nor allow the assessor to see the impact of their uncertainty on their confidence in a \emph{pfd} bound. This, despite it often being the case that an assessor may have good reason to believe in i.i.d., but not enough reason to be certain that it holds.  Furthermore, these proposals do not directly support \emph{conservative} reliability claims.  

Using conservative Bayesian techniques -- in particular, CBI -- we show how doubts about i.i.d. can be formally included in software reliability assessments (see Sections~\ref{sec_consconfbounds} and \ref{sec_results}). In this way, we obviate the need imposed by the previous proposals -- the need for assessors to either assume/conclude that the software test outcomes are i.i.d., or assume/conclude that they aren't. Instead, our method allows a skeptical assessor's confidence in a target \emph{pfd}, and their confidence in (and doubts about) i.i.d., to grow or shrink in response to seeing the software operate without failure.  

Moreover, CBI's conservative confidence bounds continuously invite our assessor to be skeptical, and to question whether seemingly favourable reliability evidence from testing does, in fact, corroborate actual reliability. For example, while failure-free operation is generally an indicator of a desirable \emph{pfd}, our results highlight why this might be (at best) a sanguine view in some situations. There may be undesirable reasons for why failure-free operation is occurring; reasons that ultimately undermine one's confidence in the \emph{pfd} being sufficiently small (see section \ref{sec_results}). CBI, by weakening the i.i.d. assumption and producing conservative confidence bounds, can call into question the representativeness of failure-free operation (as indicating a reliable system). When this happens, it's incumbent on the assessor to rule out potential problems during testing that ``masquerade'' as failure-free operation, and to incorporate these efforts into any further use of CBI. 

\section*{Acknowledgments}
We are grateful to the anonymous reviewers whose comments were very helpful in improving the presentation. We are also grateful to Bev Littlewood for giving very helpful feedback on an initial draft of this paper. This work was partly funded by the European Union’s Horizon 2020 Research and Innovation Programme under grant agreement No 956123, and by the UK EPSRC through the End-to-End Conceptual Guarding of Neural Architectures [EP/T026995/1]. Xingyu Zhao's contribution is partially supported through Fellowships at the Assuring Autonomy International Programme.

\bibliographystyle{IEEEtran}
\bibliography{ref}

\newpage

\appendices
\section{Transition Probabilities in The Klotz Model}
	\label{app_KlotzModel}
	\emph{1st-order stationarity} requires that the probability of being in a given state after $n$ trials is the same for all $n$. In particular, the probability of being in a successful state after two trials is the same as the probability after one trial, i.e. $1-x$. So, upon writing the shorthand $p=P(T_2=0\mid T_1=0)$, we have $1-x = x(1-\lambda) + (1-x)p$. Solving for $p$ gives $P(T_2=0\mid T_1=0) \ =\ p \ =\ 1-\frac{(1-\lambda)x}{1-x}$ and $P(T_2=1\mid T_1=0) \ =\ 1 - p \ =\ \frac{(1-\lambda)x}{1-x}$, for $0\leqslant x<1$.
\section{Proof of the Theorem}
\label{sec_app_C}	
\begin{proof} 
Choose any $F\in\mathcal D$ that satisfies the constraints of the optimisation and denote the Klotz likelihood \eqref{eqn_xKlotzlklhdFn_maintxt} as $L$. The objective function \eqref{eqn_Klotzmodelpostconf} in the theorem,  computed using $F$, is $\frac{\int_{\{x\leqslant b\}\cap\mathcal R} L \diff F}{\int_{\mathcal R} L \diff F} = \left(1 + \frac{\int_{\{x> b\}\cap\mathcal R} L \diff F}{\int_{\{x\leqslant b\}\cap\mathcal R} L \diff F}\right)^{-1}$.
Consequently, we focus on the equivalent optimisation (subject to the same constraints)
\begin{align}
\sup\limits_{\mathcal D} \frac{\int_{\{x > b\}\cap\mathcal R} L \diff F}{\int_{\{x\leqslant b\}\cap\mathcal R} L \diff F}
\label{eqn_equivoptimisation_sup}
\end{align} 

From $F$, one can construct a sequence of priors $\{F^\ast_k\}$ (for $k=1,2,\ldots$) that: {\bf i)} all give larger values than $F$ for the objective function in \eqref{eqn_equivoptimisation_sup}; and {\bf ii)} give objective function values that converge to the objective function value given by some $F^\ast\in\mathcal D$. The construction is as follows. Consider the sequence $\{{\mathcal P}_k\}$ of partitions of the interval $[0,1]$, defined by ${\mathcal P}_k = \{[0,\sfrac{1}{2^k}), [\sfrac{1}{2^k}, \sfrac{2}{2^k}), \ldots, [1-\sfrac{1}{2^k},1]\}$. Each partition induces a partition of $\mathcal R$ into vertical strips, as illustrated in Figure~\ref{fig:fig_proofContinuousMarg_partitions}. Within the $i$th strip, denote the region above the diagonal as $\bf r_{ia}$, the region below the diagonal as $\bf r_{ib}$, and the diagonal segment within the strip as $\bf r_{id}$. Let $i^\ast$ denote the unique index for the strip containing the vertical line $x=b$. Then, for each $F$, partition ${\mathcal P}_k$ allows the objective function in \eqref{eqn_equivoptimisation_sup} to be rewritten,
\begin{align}
\frac{\sum\limits_{i^\ast<i\leqslant 2^k}\int_{\bf r_{ia}\cup r_{ib}\cup r_{id}} L \diff F + \int_{\{x\in (b,\sfrac{i^\ast}{2^k})\}\cap\mathcal R} L \diff F}{\sum\limits_{1\leqslant i<i^\ast}\int_{\bf r_{ia}\cup r_{ib}\cup r_{id}} L \diff F + \int_{\{x\in[\sfrac{(i^\ast-1)}{2^k},b]\}\cap\mathcal R} L \diff F}
\label{eqn_objfunc_partition}    
\end{align}

$L$ is continuous and bounded over $\mathcal R$. So we may bound \eqref{eqn_objfunc_partition} from above by reallocating the probability mass that $F$ assigns within each region/diagonal segment in each strip. All of the mass is reassigned to a point in the relevant region/segment, within $\frac{1}{2^k}$ distance from where $L$ is largest (when $x>b$) or smallest (when $x\leqslant b$). These locations at which $L$ takes its largest and smallest values are \emph{limit points}\footnote{Definition: for the ``open balls'' topology associated with the 2D Euclidean plane, a \emph{limit point} of a subset of the plane is a point that is arbitrarily well-approximated by sequences of points within the subset \cite{rudin1976principles,bryant_1985}.} of the respective regions/diagonal segment within each strip, as illustrated in Figure~\ref{fig:fig_proofContinuousMarg_wcdiscretepriors}. The reallocations define a prior $F^\ast_k$ with a discrete marginal distribution of \emph{pfd}. For each $k$, $F_k^\ast$ satisfies $\frac{\int_{\{x > b\}\cap\mathcal R} L \diff F^\ast_k}{\int_{\{x\leqslant b\}\cap\mathcal R} L \diff F^\ast_k} > \frac{\int_{\{x > b\}\cap\mathcal R} L \diff F}{\int_{\{x\leqslant b\}\cap\mathcal R} L \diff F}$.

By construction, the objective function values from the $F_k^\ast$ converge to the objective function value for some prior $F^\ast$ with continuous marginal density $f(x)$. So, for each $F$,
\begin{align}
 \frac{\int_{\{x > b\}\cap\mathcal R} L \diff F^\ast}{\int_{\{x\leqslant b\}\cap\mathcal R} L \diff F^\ast}\geqslant\inf\limits_k\frac{\int_{\{x > b\}\cap\mathcal R} L \diff F^\ast_k}{\int_{\{x\leqslant b\}\cap\mathcal R} L \diff F^\ast_k} \geqslant \frac{\int_{\{x > b\}\cap\mathcal R} L \diff F}{\int_{\{x\leqslant b\}\cap\mathcal R} L \diff F}
\label{eqn_worseningpriors}    
\end{align}
Since this holds for any feasible prior $F\in\mathcal D$, we have
\begin{align}
\sup\limits_{{\mathcal D}^\ast}\frac{\int\limits_{\{x > b\}\cap\mathcal R}\!\!\!\!\!\! L \diff F^\ast}{\int\limits_{\{x\leqslant b\}\cap\mathcal R}\!\!\!\!\!\! L \diff F^\ast}\geqslant\sup\limits_{\mathcal D}\inf\limits_k\frac{\int\limits_{\{x > b\}\cap\mathcal R} \!\!\!\!\!\!L \diff F^\ast_k}{\int\limits_{\{x\leqslant b\}\cap\mathcal R} \!\!\!\!\!\!L \diff F^\ast_k} \geqslant \sup\limits_{\mathcal D}\frac{\int\limits_{\{x > b\}\cap\mathcal R} \!\!\!\!\!\!L \diff F}{\int\limits_{\{x\leqslant b\}\cap\mathcal R} \!\!\!\!\!\!L \diff F}
\label{eqn_worseningpriors_2}    
\end{align}
where ${\mathcal D}^\ast$ contains all of the $F^\ast$ priors. Because the objective function values for priors in ${\mathcal D}^\ast$ are the limits of objective function values for feasible priors in $\mathcal D$, we also have
\begin{align}
\sup\limits_{{\mathcal D}^\ast}\frac{\int_{\{x > b\}\cap\mathcal R} L \diff F^\ast}{\int_{\{x\leqslant b\}\cap\mathcal R} L \diff F^\ast} \leqslant \sup\limits_{\mathcal D}\frac{\int_{\{x > b\}\cap\mathcal R} L \diff F}{\int_{\{x\leqslant b\}\cap\mathcal R} L \diff F}
\label{eqn_worseningpriors_3}    
\end{align}
Thus \eqref{eqn_worseningpriors_2} and \eqref{eqn_worseningpriors_3} imply three equivalent forms of optimisation,
\begin{align}
\sup\limits_{{\mathcal D}^\ast}\frac{\int\limits_{\{x > b\}\cap\mathcal R} \!\!\!\!\!\!L \diff F^\ast}{\int\limits_{\{x\leqslant b\}\cap\mathcal R} \!\!\!\!\!\!L \diff F^\ast}=\sup\limits_{\mathcal D}\inf\limits_k\frac{\int\limits_{\{x > b\}\cap\mathcal R} \!\!\!\!\!\!L \diff F^\ast_k}{\int\limits_{\{x\leqslant b\}\cap\mathcal R} \!\!\!\!\!\!L \diff F^\ast_k} = \sup\limits_{\mathcal D}\frac{\int\limits_{\{x > b\}\cap\mathcal R} \!\!\!\!\!\!L \diff F}{\int\limits_{\{x\leqslant b\}\cap\mathcal R} \!\!\!\!\!\!L \diff F}
\label{eqn_equivoptimisation_Dstar}
\end{align}
So, we can restrict the optimisation to sequences of priors $\{F^\ast_k\}$.

\begin{figure}[h!]
\captionsetup[figure]{format=hang}	
	\begin{subfigure}[]{0.48\linewidth}
	\centering
	\includegraphics[width=1.0\linewidth]{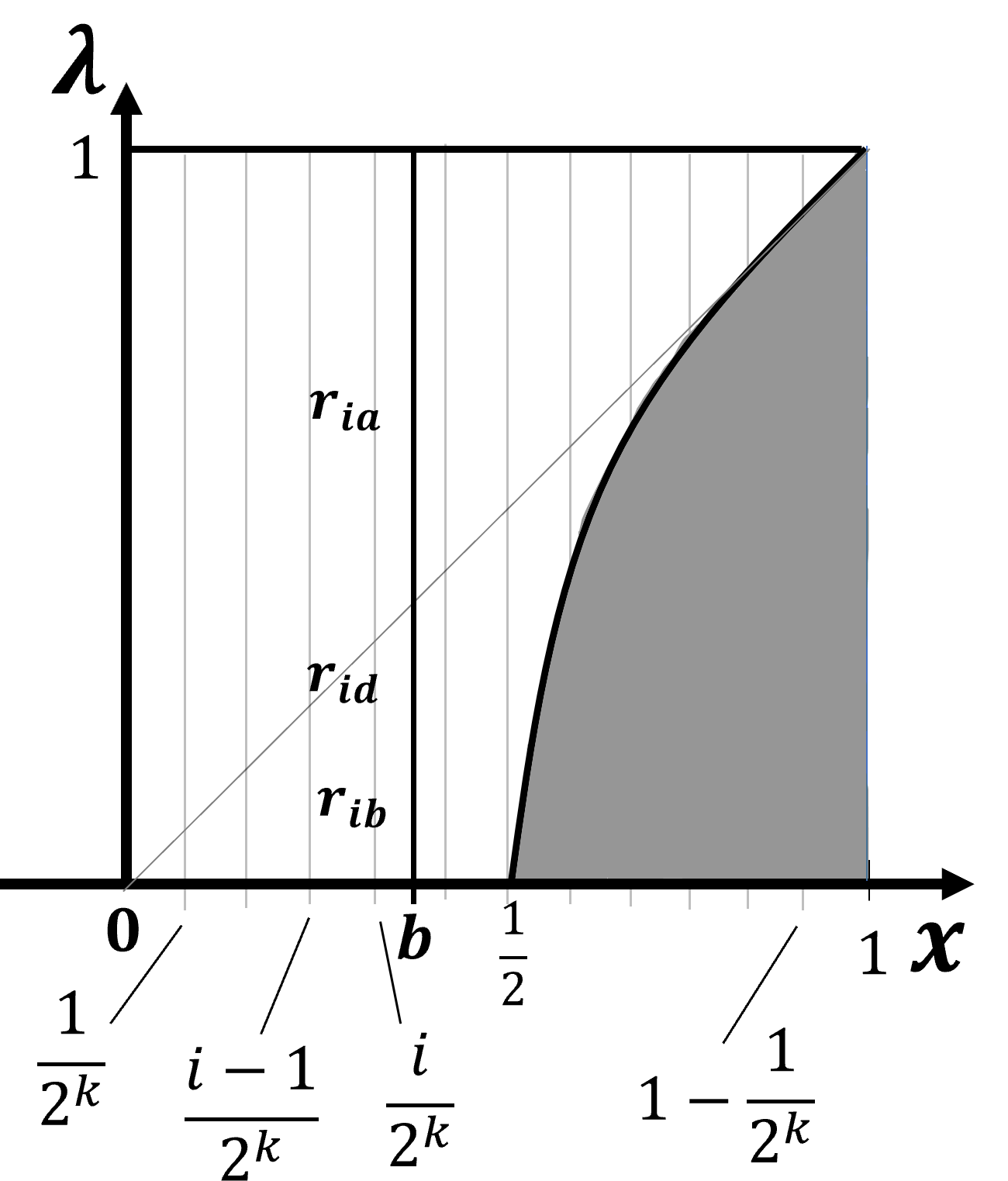}
	\caption{{\footnotesize \normalsize}}
	\label{fig:fig_proofContinuousMarg_partitions}
	\end{subfigure}
	\begin{subfigure}[h!]{0.48\linewidth}
	\centering
	\includegraphics[width=1.0\linewidth]{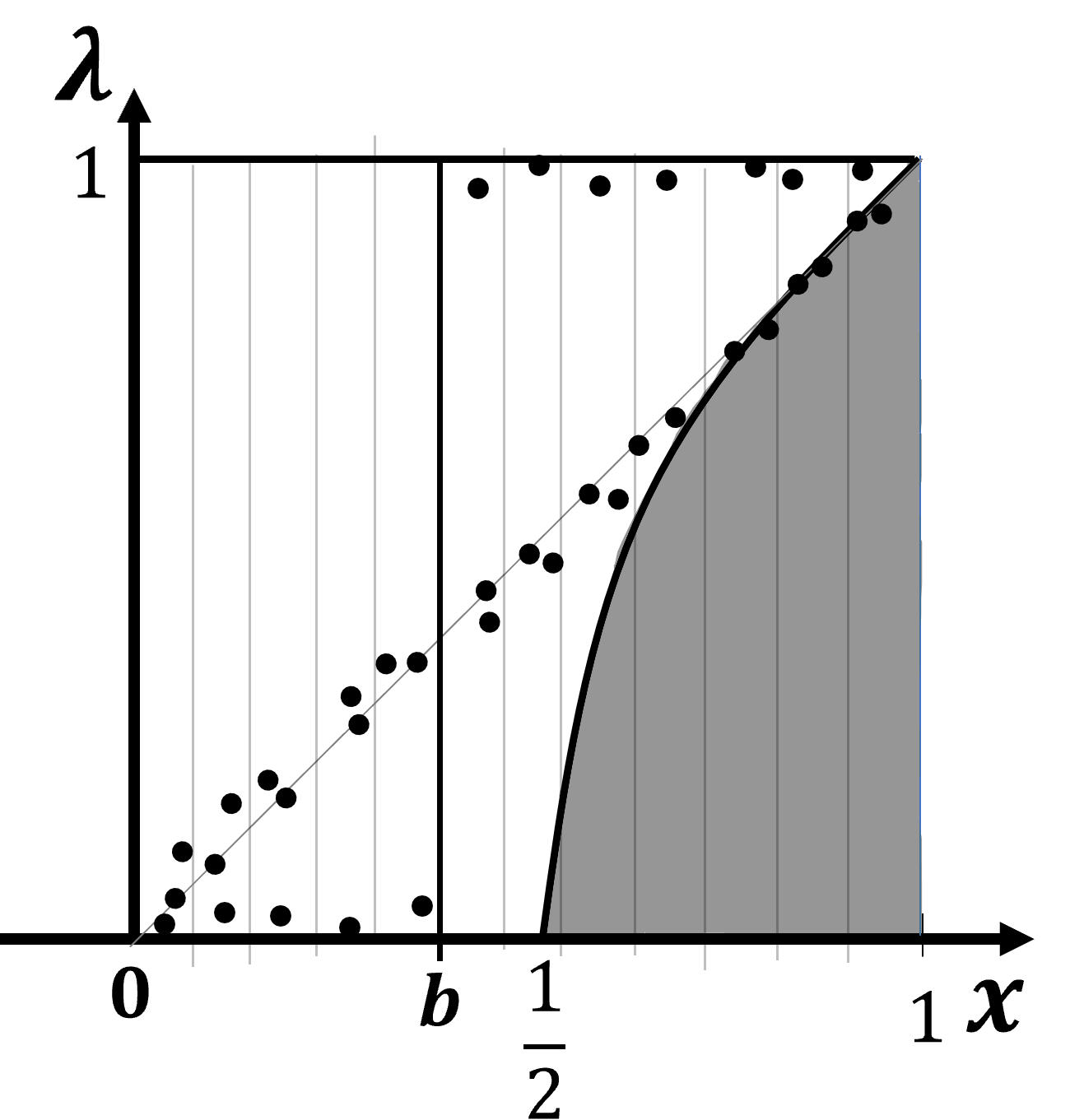}
	\vspace{0.15cm}
	\caption{{\footnotesize \normalsize}}	
	\label{fig:fig_proofContinuousMarg_wcdiscretepriors}
	\end{subfigure}
\caption[partition of {\mathcal R}]{{\footnotesize {\bf (a)} ${\mathcal P}_k$ partitions $\mathcal R$ into vertical strips; {\bf (b)} example support of $F^\ast_k$  \normalsize}}  
\label{fig:fig_proofContinuousMarg}
\end{figure}

For all sufficiently large $k$, the width of the strips can be made as small as we please. Consequently, by considering sufficiently large $k$, we may treat the location of masses within each strip as lying on the same vertical line, with masses on the diagonal segment or on the $\lambda=0,1$ borders of $\mathcal R$. Consider then an arbitrary prior $F^\ast_k$ (with discrete marginal) for sufficiently large $k$. The probability masses in a pair of strips can be reallocated within each strip to construct a new prior that gives a larger objective function value. One does this as follows. 

Let the functions $g_l(x)$ and $g_u(x)$ be as defined in \eqref{eqn_thegfunctions}. Denote the unique $x$ values at which $g_l(x)$ and $g_u(x)$ attain their maxima as $x_l$ and $x^u$, respectively. There are 4 possibilities for reallocating probability masses, based on the relative sizes of $x_l$, $x^u$ and $b$.


\begin{figure}[htbp!]
\captionsetup[figure]{format=hang}	
	\begin{subfigure}[]{0.48\linewidth}
	\centering
	\includegraphics[width=1.0\linewidth]{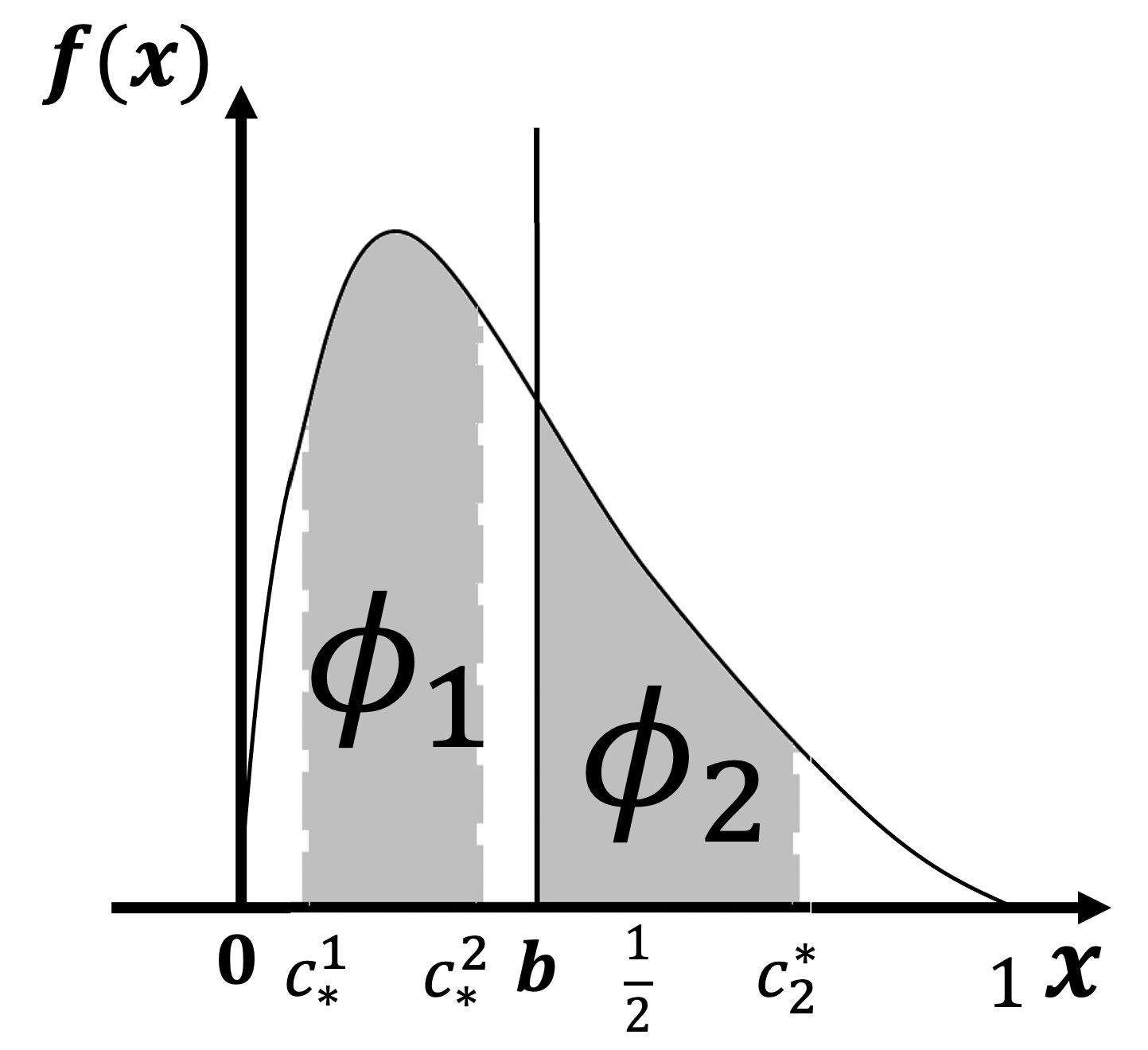}
	\caption{{\footnotesize \normalsize}}
	\label{fig:fig_density_f}
	\end{subfigure}
	\begin{subfigure}[h!]{0.48\linewidth}
	\centering
	\includegraphics[width=1.0\linewidth]{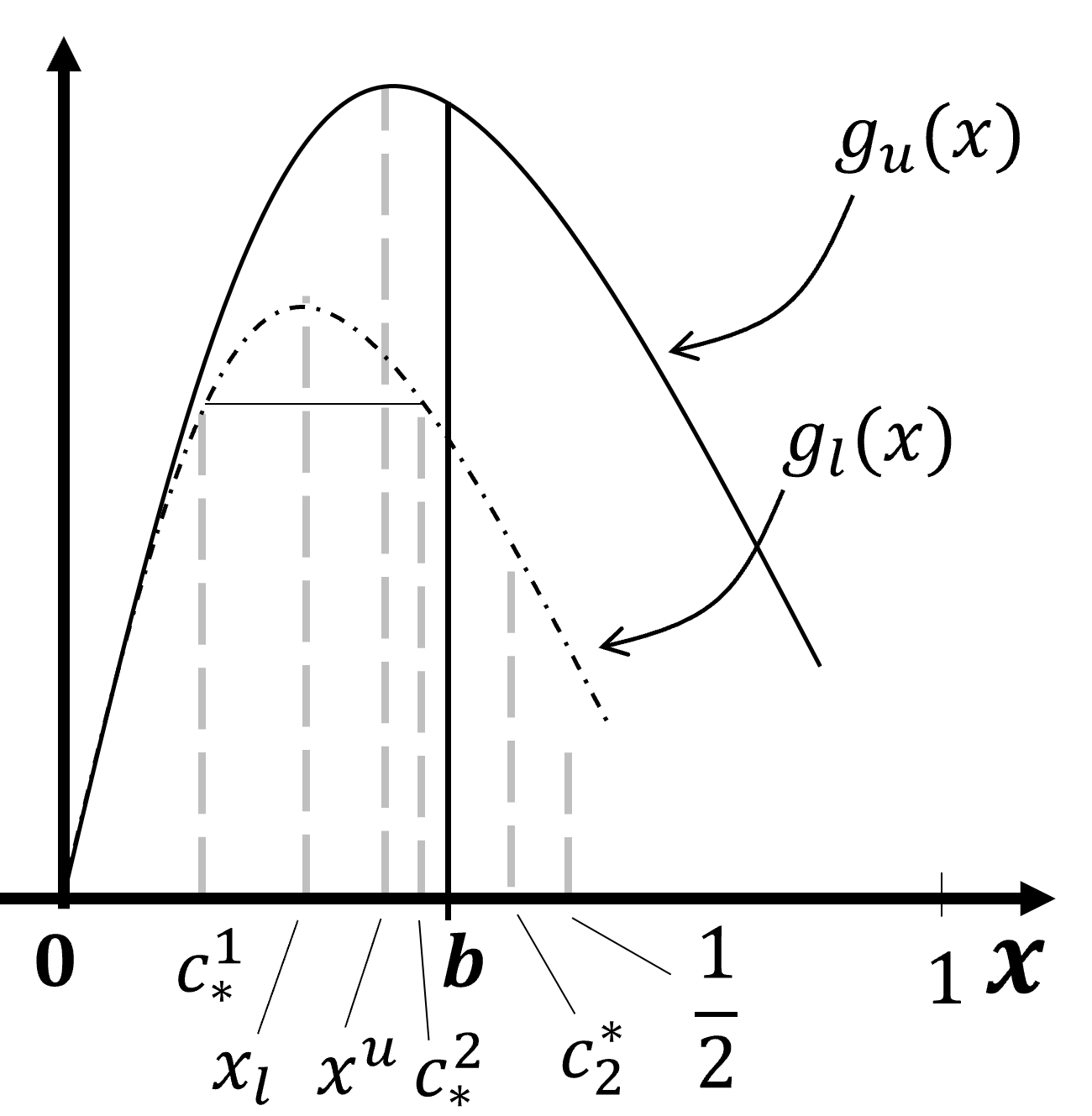}
	\caption{{\footnotesize \normalsize}}	
	\label{fig:fig_gfunctions}
	\end{subfigure}
\caption[Continuous Marginal Constraints]{{\footnotesize {\bf (a)} An example density $f(x)$; {\bf (b)} example maxima for $g_l$, $g_u$  \normalsize}}  
\label{fig:fig_continuous_marginal_constraints}
\end{figure}


\subsubsection*{ {\bf\emph{Case 1)}} $x_l<b$ and $x^u<b$} 

Let $F^\ast_k$ be as depicted in Figure~\ref{fig:fig_proofContinuousMarg_wcdiscretepriors}. Consider two vertical strips, as shown in Figure~\ref{fig:fig_movemass_c1_star}. The strips lie to the left of the vertical line $x=x_l$. For $\Delta=0$, let the probabilities $M_1-\Delta$, $M_2+\Delta$, $M_3+\Delta$ and $M_4-\Delta$ be initially assigned to the 4 depicted locations (2 in each strip). These ``$M$''s are constant and consistent with the PKs, and $\Delta$ is a sufficiently small probability mass. The derivative of the objective function with respect to $\Delta$ exists, because the objective function is a rational function of $\Delta$. The sign of this derivative is determined by the function $g_l(x)$ in Figure~\ref{fig:fig_gfunctions}. That is, the expression for the derivative is negative \emph{iff} $g_l(x_2)-g_l(x_1)>0$ (where, for $x_1<x_2<x_l<b$, $x_1$ is in the leftmost strip and $x_2$ is in the other strip). But this is true because $g_l(x)$ is unimodal.

\begin{figure}[h!]
\captionsetup[figure]{format=hang}	
	\begin{subfigure}[]{0.48\linewidth}
	\centering
	\includegraphics[width=1.0\linewidth]{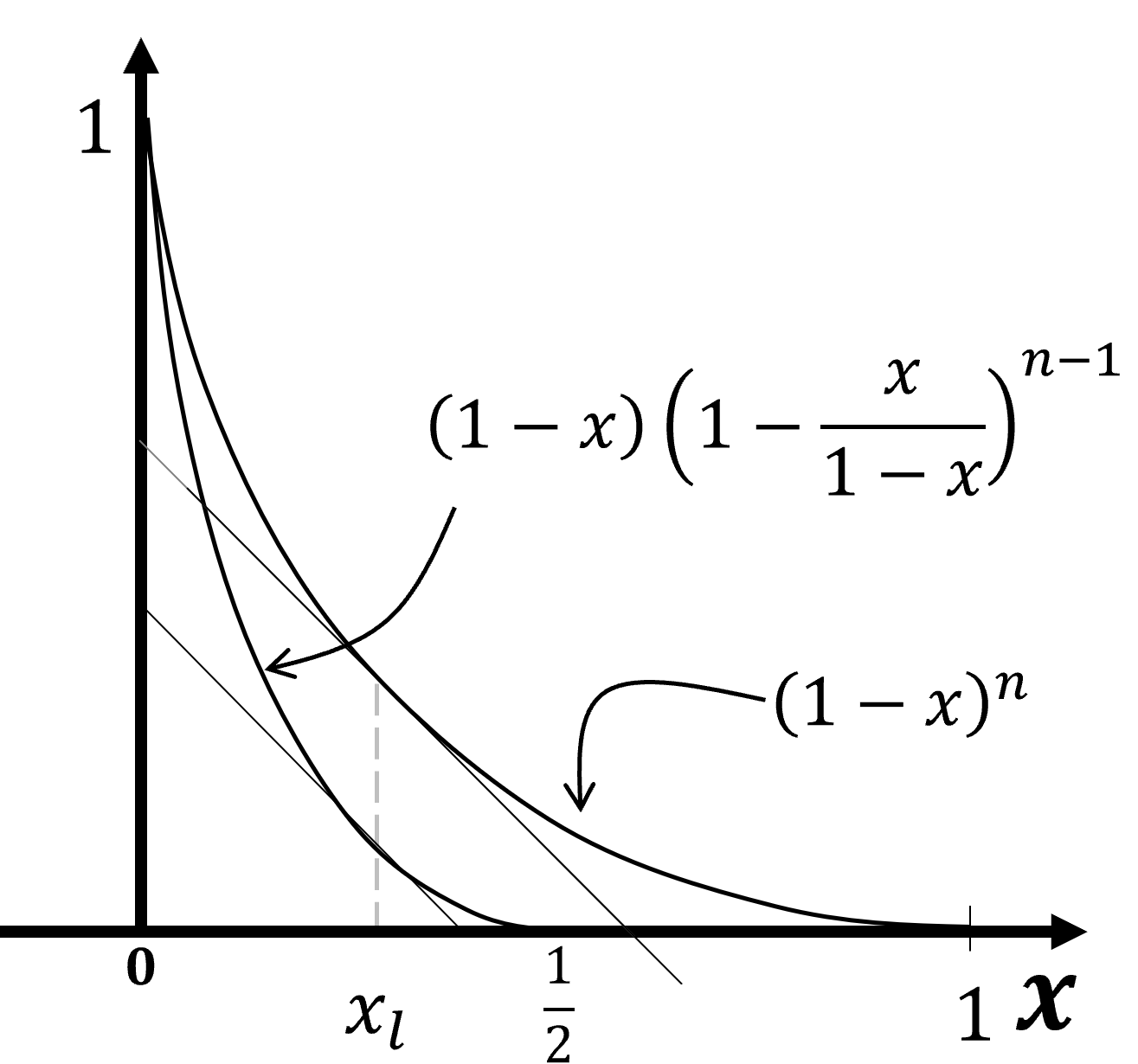}
	\caption{{\footnotesize \normalsize}}
	\label{fig:fig_glowerUnimodality}
	\end{subfigure}
	\begin{subfigure}[h!]{0.48\linewidth}
	\centering
	\includegraphics[width=1.0\linewidth]{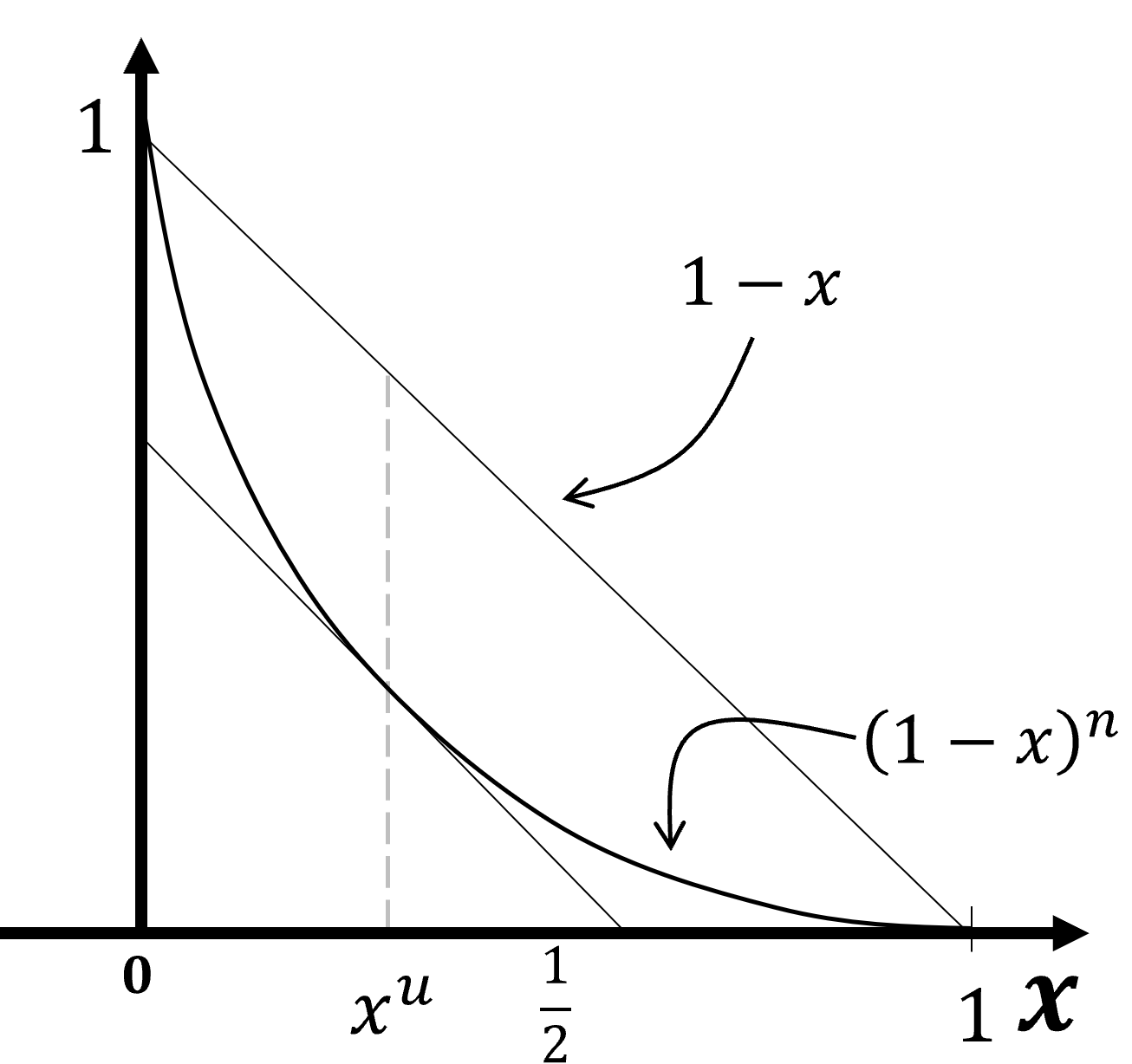}
	\caption{{\footnotesize \normalsize}}	
	\label{fig:fig_gupperUnimodality}
	\end{subfigure}
\caption[Unimodality of g functions]{{\footnotesize Plots of the pair of functions that define {\bf (a)} $g_l$,\, {\bf (b)} $g_u$  \normalsize}}  
\label{fig:fig_unimodalityofgfunctions}
\end{figure}

The unimodality of both $g_l(x)$ and $g_u(x)$ can be seen from arguments illustrated by Figures~\ref{fig:fig_glowerUnimodality} and \ref{fig:fig_gupperUnimodality}. These figures depict the pair of functions that define each ``$g$'' function. Each pair consists of two convex, monotonically decreasing functions that are equal at $x=0$. In Figure~\ref{fig:fig_glowerUnimodality}, over $0\leqslant x\leqslant\frac{1}{2}$, the pair of functions are initially relatively convex (with relative derivative 1 at $x=0$) then relatively concave (with relative derivative 0 at $x=\frac{1}{2}$). While in Figure~\ref{fig:fig_gupperUnimodality}, the pair of functions are relatively concave over $0\leqslant x\leqslant\frac{1}{2}$ (with relative  derivatives $n$ and $n(\frac{1}{2})^{n-1}$, at $x=0$ and $x=\frac{1}{2}$ respectively). Because of these, in each figure, the functions have the same tangent slope at a nontrivial $x$ value in their shared domain. This is the $x$ value at which the respective ``$g$'' function attains its maximum -- the values $x_l$ and $x^u$. These values lie in the interval $0\leqslant x< \sfrac{1}{2}$.




Since the objective function's derivative with respect to $\Delta$ is negative, $\Delta$ should be made as small as possible, which makes the objective function as large as possible. Roughly speaking, mass in the ``$x_1$'' strip should be placed on the diagonal, while mass in the ``$x_2$'' strip should be placed along the $\lambda=0$ line. Similar arguments justify the mass re-allocations illustrated in Figures~\ref{fig:fig_movemass_c2_star} and \ref{fig:fig_movemass_cstar}, using $g_l(x)$ and $g_u(x)$ respectively.

The general rule is, for a pair of strips containing $x$ values less than $b$, the strip that is closest to containing $x_l$ should have as much mass as possible below the diagonal, while the other strip should have as much mass as possible on or above the diagonal. Similarly, for two strips with $x$ values greater than $b$, the strip closest to containing $x^u$ should have as much mass as possible above the diagonal, while the other strip should have as much mass as possible on or below the diagonal. 

So, by construction, a discrete prior $F^{\ast}_k$ (e.g. Figure~\ref{fig:fig_proofContinuousMarg_wcdiscretepriors}) is replaced by a more extreme $F^{\ast\ast}_k$ (e.g. Figure~\ref{fig:fig_continousmarginal_discreteapprx_phi1phi2small}). Further reallocation is impossible when $c^1_\ast$, $c^2_\ast$, $c_1^\ast$ and $c_2^\ast$ have been found that solve
\[\begin{array}{ccc}
&\argmin\limits_{0\leqslant r<s\leqslant b}\mid g_l(r)-g_l(s)\mid,\,\argmin\limits_{b\leqslant v<w\leqslant 1}\mid g_u(v)-g_u(w)\mid&\\
\mbox{s.t.} \,\,\,\,\,
&g_l(0)\leqslant g_l(r),\,\,\,\,g_l(b)\leqslant g_l(s)\,,&\\
&g_u(b) \leqslant g_u(v),\,\,\,\,g_u(1)\leqslant g_u(w)\,,&\\ 
&\int_{\{x\in[r,s]\}\cap\mathcal R}\diff F^{\ast\ast}_k = \int_{r}^{s}f(x)\diff x = \phi_1\,,&\\
&\int_{\{x\in[b,w]\}\cap\mathcal R}\diff F^{\ast\ast}_k =\int_{b}^{w}f(x)\diff x = \phi_2\,,&\\
&0<r<x_l<s\leqslant b\,,\,\,\,\,0<x^u<b \leqslant v<w\leqslant 1&
\end{array}\]

In particular, since $x^u<b$ implies $c^\ast_1=b$, we can restrict the optimisation to these more extreme priors $F^{\ast\ast}_k$. For such priors, the objective function \eqref{eqn_objfunc_partition} is comprised of sums that are integrals (with respect to $F^{\ast\ast}_k$) of simple functions. That is:
\begin{align*}
&\frac{\sum\limits_{i^\ast<i\leqslant 2^k}\int_{\bf r_{ia}\cup r_{ib}\cup r_{id}} L \diff F^{\ast\ast}_k + \int_{\{x\in (b,\sfrac{2^{i^\ast}}{2^k})\}\cap\mathcal R} L \diff F^{\ast\ast}_k}{\sum\limits_{1\leqslant i<i^\ast}\int_{\bf r_{ia}\cup r_{ib}\cup r_{id}} L \diff F^{\ast\ast}_k + \int_{\{x\in[\sfrac{2^{i^\ast-1}}{2^k},b]\}\cap\mathcal R} L \diff F^{\ast\ast}_k} \\
{}={}&\frac{\int_{\{x\in(b,c_2^\ast)\}\cap\mathcal R} L \diff F^{\ast\ast}_k + \int_{\{x\in (c_2^\ast,1)\}\cap\mathcal R} L \diff F^{\ast\ast}_k}{\int_{\{x\in(c^1_\ast,c^2_\ast)\}\cap\mathcal R} L \diff F^{\ast\ast}_k + \int_{\{x\in(0,c^1_\ast)\cup(c^2_\ast,b)\}\cap\mathcal R} L \diff F^{\ast\ast}_k}\\
{}={}&\frac{\sum\limits_{i^{\ast}<i\leqslant 2^k}\!\!\!\left(L_1{\bf 1}_{x_i\in(b,c^{\ast}_2)} + L_3{\bf 1}_{x_i\in(c^{\ast}_2,1)}\right)\!\!\!\int\limits_{\mbox{\scriptsize ith strip}}\!\!\!\!\!\diff F^{\ast\ast}_k}{\sum\limits_{1\leqslant i \leqslant i^{\ast}}\!\!\!\left(L_2{\bf 1}_{x_i\in(c_{\ast}^1,c_{\ast}^2)} + L_3{\bf 1}_{x_i\in(0,c_{\ast}^1)\cup(c_{\ast}^2,b)}\right)\!\!\!\int\limits_{\mbox{\scriptsize ith strip}}\!\!\!\!\!\diff F^{\ast\ast}_k} \\
{}={}&\frac{\sum\limits_{i^{\ast}<i\leqslant 2^k}\!\!\!\left(L_1{\bf 1}_{x_i\in(b,c^{\ast}_2)} + L_3{\bf 1}_{x_i\in(c^{\ast}_2,1)}\right)\!\!\int\limits_{\frac{i-1}{2^k}}^{\frac{i}{2^k}}\!\!f(x)\diff x}{\sum\limits_{1\leqslant i \leqslant i^{\ast}}\!\!\!\left(L_2{\bf 1}_{x_i\in(c_{\ast}^1,c_{\ast}^2)} + L_3{\bf 1}_{x_i\in(0,c_{\ast}^1)\cup(c_{\ast}^2,b) }\right)\!\!\int\limits_{\frac{i-1}{2^k}}^{\frac{i}{2^k}}\!\!f(x)\diff x}
\end{align*}
where $x_i\in[\frac{i-1}{2^k},\frac{i}{2^k})$, $L_1:=(1-x_i)$, $L_2:=\frac{(1-2x_i)^{n-1}}{(1-x_i)^{n-2}}$ and $L_3:=(1-x_i)^n$

\begin{figure}[htbp!]
\captionsetup[figure]{format=hang}	
	\begin{subfigure}[]{0.51\linewidth}
	\centering
	\includegraphics[width=1.0\linewidth]{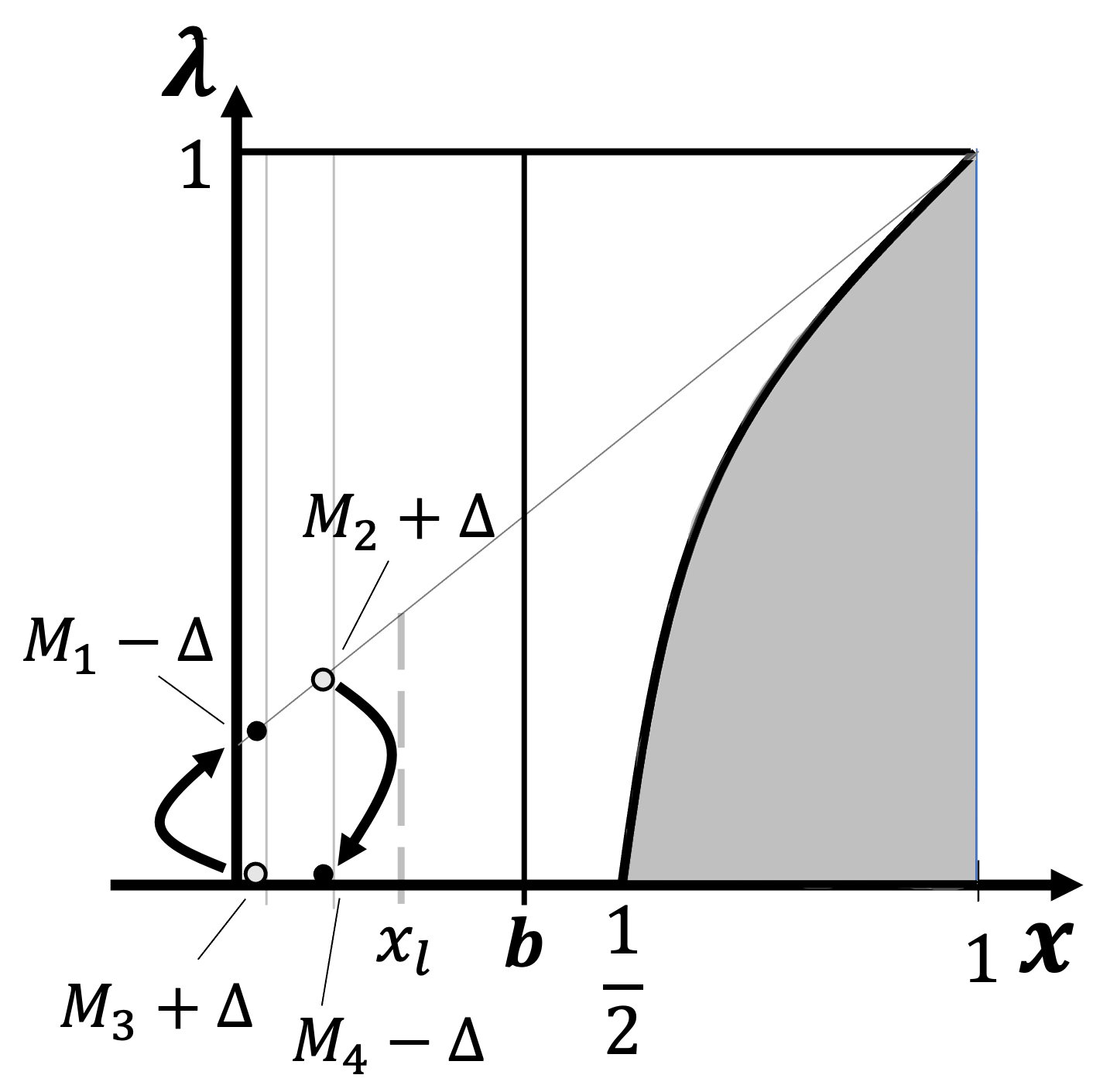}
	\caption{{\footnotesize \normalsize}}	
	\label{fig:fig_movemass_c1_star}
	\end{subfigure}	
	\begin{subfigure}[]{0.47\linewidth}
	\centering
	\includegraphics[width=1.0\linewidth]{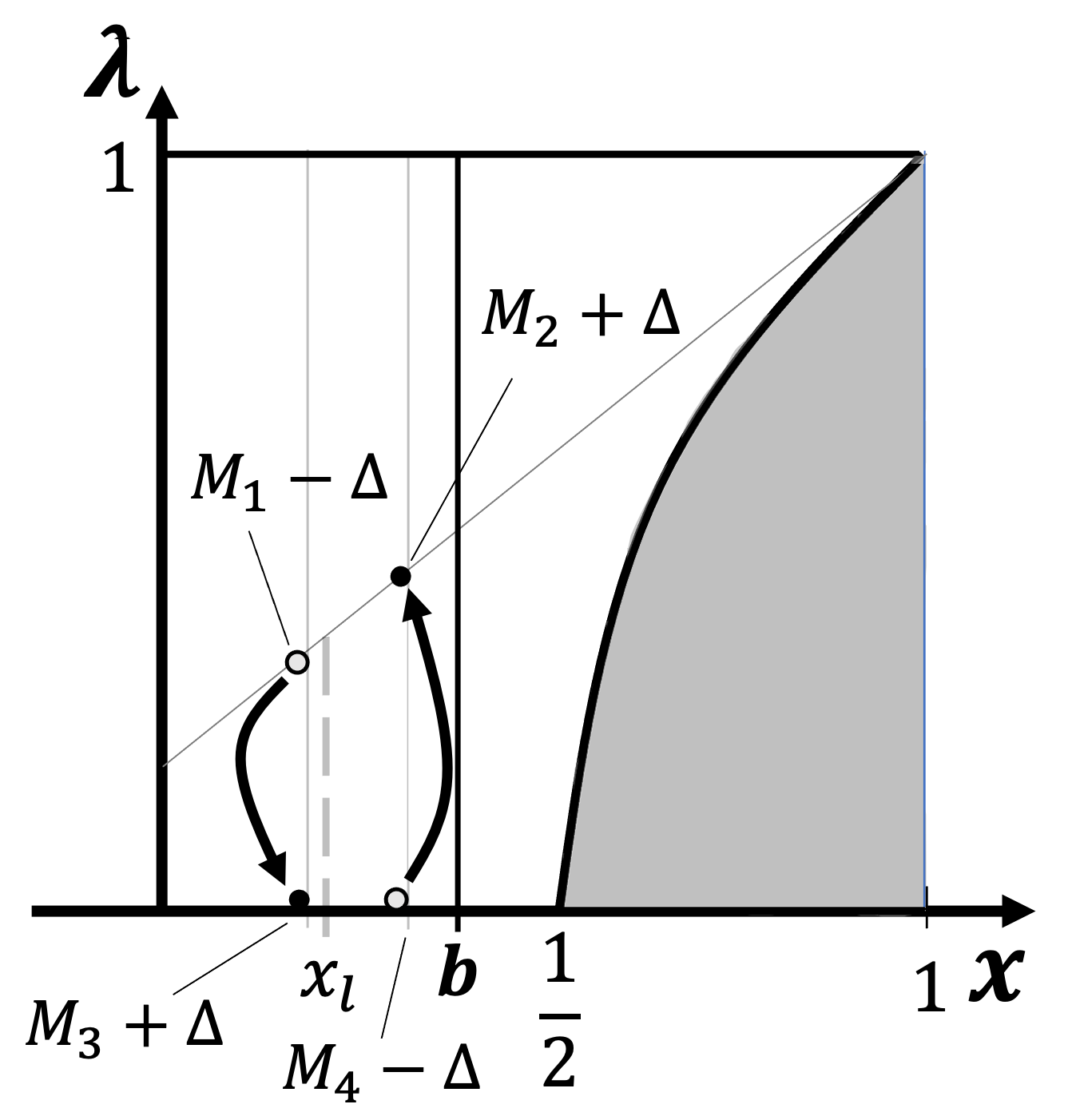}
	\caption{{\footnotesize \normalsize}}	
	\label{fig:fig_movemass_c2_star}
	\end{subfigure}	
	\begin{subfigure}[]{0.48\linewidth}
	\centering
	\includegraphics[width=1.0\linewidth]{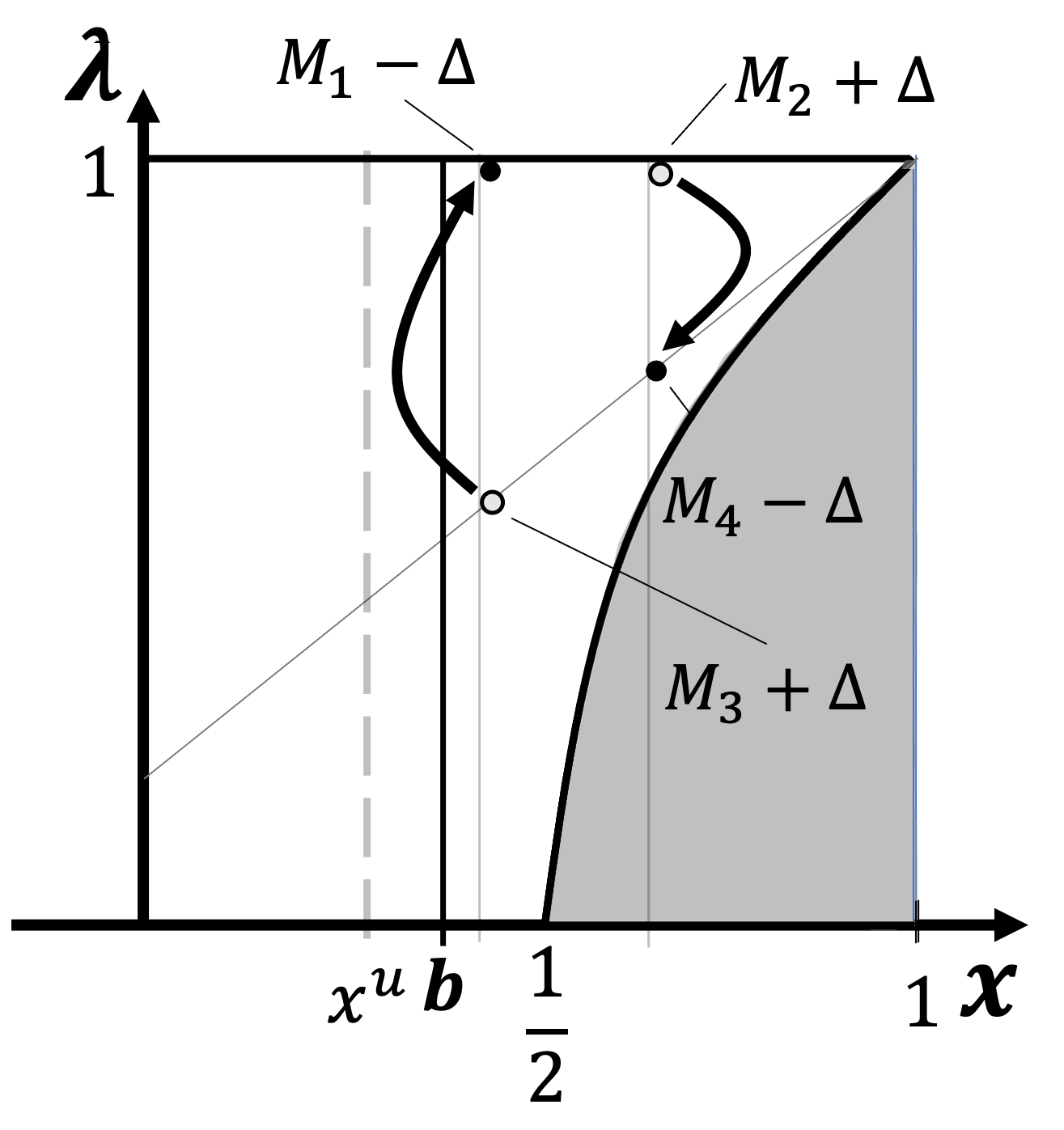}
	\caption{{\footnotesize \normalsize}}	
	\label{fig:fig_movemass_cstar}
	\end{subfigure}	
	\begin{subfigure}[]{0.48\linewidth}
	\centering
	\includegraphics[width=1.0\linewidth]{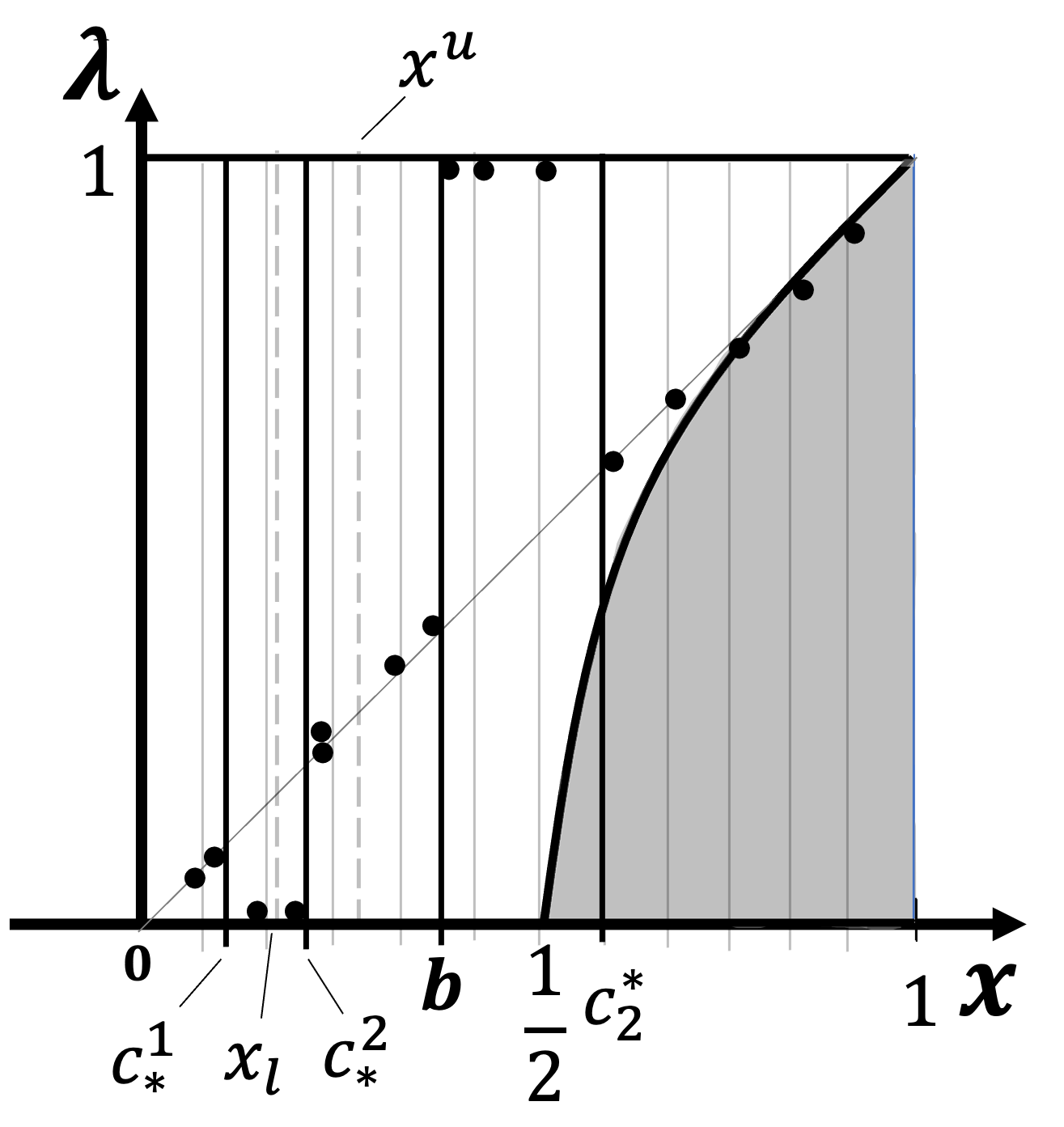}
	\caption{{\footnotesize \normalsize}}
	\label{fig:fig_continousmarginal_discreteapprx_phi1phi2small}
	\end{subfigure}
\caption[Mass movement in Discrete Priors]{{\footnotesize $F^{\ast\ast}_k$ is constructed from $F^\ast_k$ by reallocating probability mass   \normalsize}}  
\label{fig_}
\end{figure}

By the \emph{dominated convergence theorem} (see \cite{schilling_2005}), the continuity of $L$ over $\mathcal R$ implies that the sums converge to integrals with respect to the density $f(x)$ as $k\rightarrow\infty$, so 
\begin{align*}
\frac{\int_{b}^{1}\!\left((1-x){\bf 1}_{x\in(b,c^{\ast}_2)} + (1-x)^n{\bf 1}_{x\in(c^{\ast}_2,1)}\right)f(x)\diff x}{\int_{0}^{b}\!\left(\frac{(1-2x)^{n-1}}{(1-x)^{n-2}}{\bf 1}_{x\in(c_{\ast}^1,c_{\ast}^2)} + (1-x)^n{\bf 1}_{x\in(0,c_{\ast}^1)\cup (c_{\ast}^2,b)}\right)f(x)\diff x}\,\, .
\end{align*}

\subsubsection*{{\bf\emph{Case 2)}} $x_l<b$ and $x^u>b$}

An analogous argument to case 1 gives the solution

\begin{align*}
\frac{\int_{b}^{1}\!\left((1-x){\bf 1}_{x\in(c^\ast_1,c^{\ast}_2)} + (1-x)^n{\bf 1}_{x\in(b,c^{\ast}_1)\cup(c^{\ast}_2,1)}\right)f(x)\diff x}{\int_{0}^{b}\!\left(\frac{(1-2x)^{n-1}}{(1-x)^{n-2}}{\bf 1}_{x\in(c_{\ast}^1,c_{\ast}^2)} + (1-x)^n{\bf 1}_{x\in(0,c_{\ast}^1)\cup(c_{\ast}^2,b)}\right)f(x)\diff x}
\end{align*}
where $c^1_\ast$, $c^2_\ast$, $c_1^\ast$ and $c_2^\ast$ have been identified that solve
\[\begin{array}{ccc}
&\argmin\limits_{0\leqslant r<s\leqslant b}\mid g_l(r)-g_l(s)\mid,\,\argmin\limits_{b\leqslant v<w\leqslant 1}\mid g_u(v)-g_u(w)\mid&\\
\mbox{s.t.} \,\,\,\,\,\,\,&g_l(0)\leqslant g_l(r),\,\,\,\,g_l(b)\leqslant g_l(s)\,,&\\
&g_u(b) \leqslant g_u(v),\,\,\,\,g_u(1)\leqslant g_u(w)\,,&\\ 
&\int_{\{x\in[r,s]\}\cap\mathcal R}\diff F^{\ast\ast}_k = \int_{[r,s]}f(x)\diff x = \phi_1\,,&\\
&\int_{\{x\in[v,w]\}\cap\mathcal R}\diff F^{\ast\ast}_k =\int_{[v,w]}f(x)\diff x = \phi_2\,,&\\
&0<r<x_l<s\leqslant b\,,\,\,\,\,0<b \leqslant v<x^u<w\leqslant 1
\end{array}\]

\subsubsection*{{\bf\emph{Case 3)}} $x_l>b$ and $x^u<b$}  

An analogous argument to case 1 gives the solution

\begin{align*}
\frac{\int_{b}^{1}\!\left((1-x){\bf 1}_{x\in(b,c^{\ast}_2)} + (1-x)^n{\bf 1}_{x\in(c^{\ast}_2,1)}\right)f(x)\diff x}{\int_{0}^{b}\!\left(\frac{(1-2x)^{n-1}}{(1-x)^{n-2}}{\bf 1}_{x\in(c_{\ast}^1,b)} + (1-x)^n{\bf 1}_{x\in(0,c_{\ast}^1)}\right)f(x)\diff x}
\end{align*}
where $c^1_\ast$, $c^2_\ast$, $c_1^\ast$ and $c_2^\ast$ have been identified that solve
\[\begin{array}{ccc}
&\argmin\limits_{0\leqslant r<s\leqslant b}\mid g_l(r)-g_l(s)\mid,\,\argmin\limits_{b\leqslant v<w\leqslant 1}\mid g_u(v)-g_u(w)\mid&\\
\mbox{s.t.} \,\,\,&g_l(0)\leqslant g_l(r),\,\,\,\,g_l(b)\leqslant g_l(s)\,,&\\
&g_u(b) \leqslant g_u(v),\,\,\,\,g_u(1)\leqslant g_u(w)\,,&\\ 
&\int_{\{x\in[r,b]\}\cap\mathcal R}\diff F^{\ast\ast}_k = \int_{r}^{b}f(x)\diff x = \phi_1\,,&\\
&\int_{\{x\in[b,w]\}\cap\mathcal R}\diff F^{\ast\ast}_k =\int_{b}^{w}f(x)\diff x = \phi_2\,,&\\
&0<r<s\leqslant b<x_l\,,\,\,\,\,0<x^u<b \leqslant v<w\leqslant 1&
\end{array}\]
In particular, because $x_l>b$, $x^u<b$, we have $c^2_\ast=b$, $c^\ast_1=b$.

\subsubsection*{{\bf\emph{Case 4)}} $x_l>b$ and $x^u>b$}  

An analogous argument to case 1 gives the solution

\begin{align*}
\frac{\int_{b}^{1}\!\left((1-x){\bf 1}_{x\in(c^\ast_1,c^{\ast}_2)} + (1-x)^n{\bf 1}_{x\in(b,c^{\ast}_1)\cup(c^{\ast}_2,1) }\right)f(x)\diff x}{\int_{0}^{b}\!\left(\frac{(1-2x_i)^{n-1}}{(1-x_i)^{n-2}}{\bf 1}_{x\in(c_{\ast}^1,b)} + (1-x)^n{\bf 1}_{x\in(0,c_{\ast}^1)}\right)f(x)\diff x}
\end{align*}
where $c^1_\ast$, $c^2_\ast$, $c_1^\ast$ and $c_2^\ast$ have been identified that solve
\[\begin{array}{ccc}
&\argmin\limits_{0\leqslant r<s\leqslant b}\mid g_l(r)-g_l(s)\mid,\,\argmin\limits_{b\leqslant v<w\leqslant 1}\mid g_u(v)-g_u(w)\mid&\\
\mbox{s.t.} \,\,\,&g_l(0)\leqslant g_l(r),\,\,\,\,g_l(b)\leqslant g_l(s)\,,&\\
&g_u(b) \leqslant g_u(v),\,\,\,\,g_u(1)\leqslant g_u(w)\,,&\\ 
&\int_{\{x\in[r,b]\}\cap\mathcal R}\diff F^{\ast\ast}_k = \int_{r}^{b}f(x)\diff x = \phi_1\,,&\\
&\int_{\{x\in[v,w]\}\cap\mathcal R}\diff F^{\ast\ast}_k =\int_{v}^{w}f(x)\diff x = \phi_2\,,& \\
&0<r<s\leqslant b<x_l\,,\,\,\,\,0<b \leqslant v<x^u<w\leqslant 1&
\end{array}\]
In particular, because $x_l>b$, we must have $c^2_\ast=b$.

\end{proof}

\section{Asymptotics of Posterior Confidence based on Failure-free Operation}
\label{sec_app_E}	
\begin{claim*}
In the theorem, $\lim\limits_{n\rightarrow\infty}Q=\left\{\begin{array}{ll} 0,&\mbox{if }\phi_2=0 \\ \infty,&\mbox{if }\phi_2>0 \end{array}\right.$. Since the assessor's conservative posterior confidence in the bound $b$ is $\frac{1}{1+Q}$, the assessor either becomes certain that $b$ has been satisfied, or they become certain that it has not.   
\end{claim*}
\begin{proof}
We will show that 
\begin{align}
&\frac{\int_{b}^{1}\!\left((1-x)^n{\bf 1}_{x\in(b,c_1^{\ast})\cup(c_2^{\ast},1)}\ +\ (1-x) {\bf 1}_{x\in(c^{\ast}_1,c^{\ast}_2)}\right)f(x)\diff x}{\int_{0}^{b}\!\left((1-x)^n {\bf 1}_{x\in(0,c^1_{\ast})\cup(c^2_{\ast},b)}\ +\ \frac{(1-2x)^{n-1}}{(1-x)^{n-2}}{\bf 1}_{x\in(c^1_{\ast},c^2_{\ast})}\right)f(x)\diff x} \nonumber \\
&\hspace{3.5cm}\xrightarrow{n\rightarrow\infty}\left\{\begin{array}{ll} 0,&\mbox{if }\phi_2=0 \\ \infty,&\mbox{if }\phi_2>0 \end{array}\right.
\label{eqn_app_CBIsoln_continousmarginal_limits_original}
\end{align}
Since $\phi_2=0$ implies $c_1^{\ast}=c_2^{\ast}$, the \emph{lhs} of \eqref{eqn_app_CBIsoln_continousmarginal_limits_original}, i.e. $Q$, becomes
\begin{align}
\frac{\int_{b}^{1}(1-x)^n f(x)\diff x}{\int_{0}^{b}\!\left((1-x)^n {\bf 1}_{x\in(0,c^1_{\ast})\cup(c^2_{\ast},b)}\ +\ \frac{(1-2x)^{n-1}}{(1-x)^{n-2}}{\bf 1}_{x\in(c^1_{\ast},c^2_{\ast})}\right)f(x)\diff x}
\label{eqn_app_CBIsoln_continousmarginal_limits}
\end{align}

\begin{figure}[h!]
	\centering
	\includegraphics[width=0.6\linewidth]{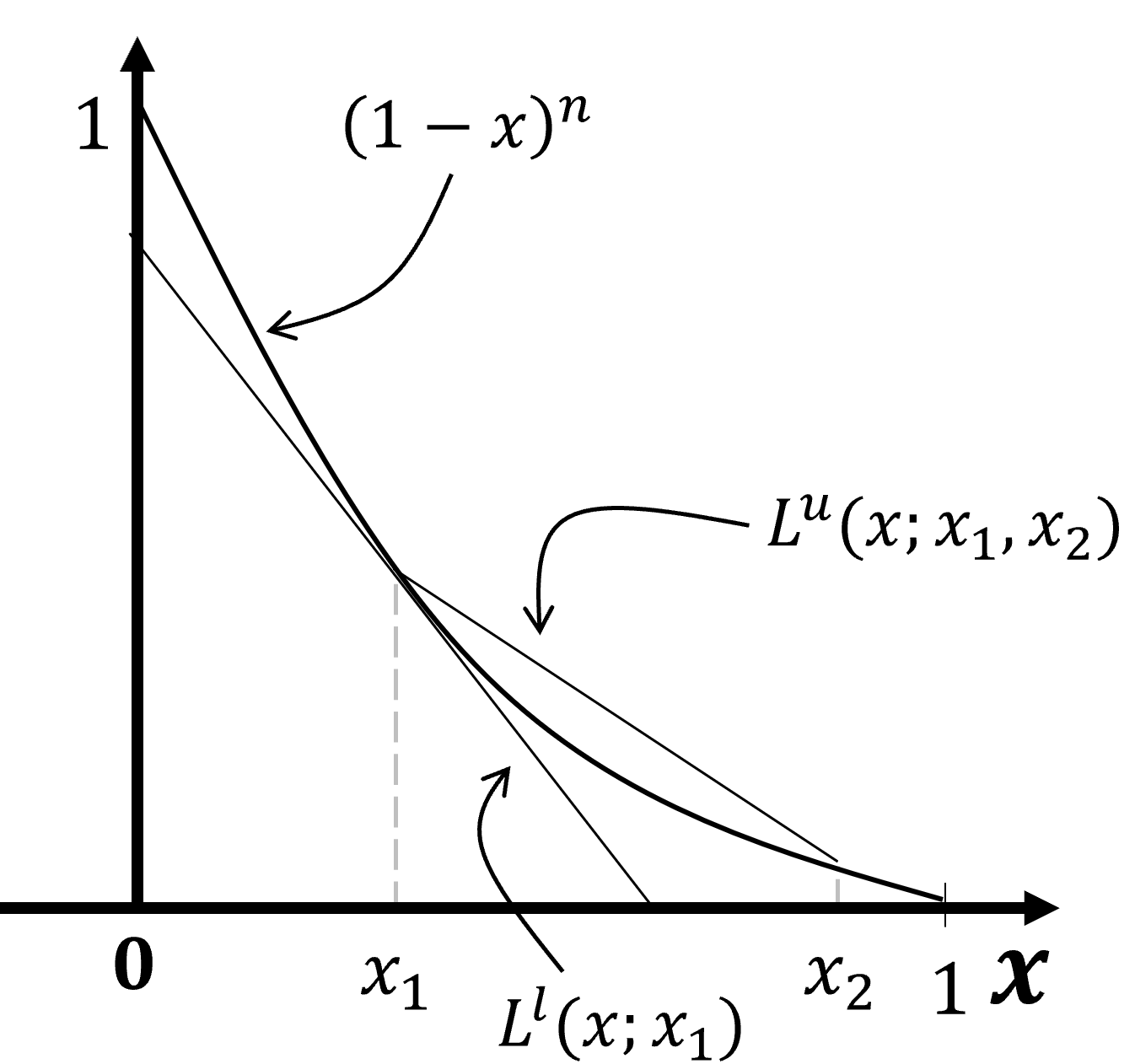}
	\caption{{\footnotesize Geometric illustration of Jensen's inequalities \normalsize}}	\label{fig:fig_jensensInequalityLines}
\end{figure}
The integrals of $(1-x)^n$ in \eqref{eqn_app_CBIsoln_continousmarginal_limits} can be bounded by a suitable choice of straight lines. We construct these as follows. For constants $x_1$ and $x_2$ such that $0\leqslant x_1<x_2\leqslant 1$, define the straight lines $L^u(x;x_1,x_2)$ and $L^l(x;x_1)$ (see Figure~\ref{fig:fig_jensensInequalityLines}):
\begin{align*}
L^u(x; x_1,x_2) &= (1-x_1)^n\!
\left(\frac{x_2-x}{x_2-x_1} \right) + (1-x_2)^n\!
\left(\frac{x-x_1}{x_2-x_1} \right) \\
L^l(x;x_1) &= \bigg(\ n(x_1-x)+1-x_1\ \bigg)(1-x_1)^{n-1}
\end{align*}
The curve $(1-x)^n$ is convex, so $L^u$ lies above the curve when $x_1<x<x_2$. While $L^l$ is tangent at $x=x_1$, so lies below the curve. These are \emph{Jensen's inequalities} \cite{Aliprantis_InfDimAnalysis1999}. Therefore: 
\begin{align}
L^l(\MyExp_{\ast}^1 X; \MyExp_{\ast}^1 X)\leqslant \frac{\int_{0}^{c^1_{\ast}}(1-x)^n f(x)\diff x}{\int_{0}^{c^1_{\ast}}\!f(x)\diff x}\leqslant L^u(\MyExp_{\ast}^1 X; 0,c^1_{\ast})
\label{eqn_boundOnIntegral_1}
\end{align}
\begin{align}
L^l(\MyExp_{\ast}^2 X; \MyExp_{\ast}^2 X)\leqslant \frac{\int_{c^2_{\ast}}^{b}(1-x)^n f(x)\diff x}{\int_{c^2_{\ast}}^{b}\!f(x)\diff x}\leqslant L^u(\MyExp_{\ast}^2 X; c^2_{\ast},b)
\label{eqn_boundOnIntegral_2}
\end{align}
\begin{align}
L^l(\MyExp X; \MyExp X)\leqslant \frac{\int_{b}^{1}(1-x)^n f(x)\diff x}{\int_{b}^{1}\!f(x)\diff x}\leqslant L^u(\MyExp X; b,1)
\label{eqn_boundOnIntegral_3}
\end{align}
where $\ \MyExp_{\ast}^1 X = \frac{\int_{0}^{c_{\ast}^1}xf(x)\diff x}{\int_{0}^{c_{\ast}^1}\!f(x)\diff x}\ $, $\MyExp_{\ast}^2 X = \frac{\int_{c_{\ast}^2}^{b}xf(x)\diff x}{\int_{c_{\ast}^2}^{b}\!f(x)\diff x}$, $\ \MyExp X = \frac{\int_{b}^{1}xf(x)\diff x}{\int_{b}^{1}\!f(x)\diff x}$.

Using the bounds in \eqref{eqn_boundOnIntegral_1}--\eqref{eqn_boundOnIntegral_3}, we can bound $Q$ (i.e. \eqref{eqn_app_CBIsoln_continousmarginal_limits}):
\begin{align}
&\hspace{3cm}0\leqslant Q\leqslant \nonumber \\
&\frac{L^u(\MyExp X;b,1)\int_{b}^{1}\!f(x)\diff x}{L^l(\MyExp_{\ast}^1 X; \MyExp_{\ast}^1 X)\!\int\limits_{0}^{c^1_{\ast}}\!f(x)\diff x\,+\,L^l(\MyExp^2_{\ast}[X];\MyExp^2_{\ast}[X])\!\int\limits_{c^2_{\ast}}^b\!f(x)\diff x} \nonumber \\
&{}={}\frac{\left(\frac{1-\MyExp X}{1-b}\right)\int_{b}^{1}\!f(x)\diff x}{\left(\frac{1-\MyExp^1_{\ast}X}{1-b}\right)^n\int_{0}^{c^1_{\ast}}\!f(x)\diff x\,+\,\left(\frac{1-\MyExp^2_{\ast}X}{1-b}\right)^n\!\int_{c^2_{\ast}}^{b}\!f(x)\diff x}
\label{eqn_boundR}
\end{align}

We used $\int_{0}^{b} \frac{(1-2x)^{n-1}}{(1-x)^{n-2}}{\bf 1}_{x\in(c^1_{\ast},c^2_{\ast})}\,f(x)\diff x> 0 $ to bound $Q$ from above -- by removing this term from $Q$'s denominator. Since $0<\MyExp^2_{\ast}X< b < 1$, we have $\left(\frac{1-\MyExp^2_{\ast}X}{1-b}\right)>1$. So, as $n\rightarrow\infty$ in \eqref{eqn_boundR}, $c^2_\ast$ tends to a non-zero value less than $b$, $\int_{c^2_{\ast}}^{b}\!f(x)\diff x$ tends to a non-zero value less than 1, $\MyExp^2_{\ast}X$ tends to a non-zero value less than $b$, $\left(\frac{1-\MyExp^2_{\ast}X}{1-b}\right)^n$ tends to $\infty$, and $\lim\limits_{n\rightarrow\infty}Q=0$.

If instead, $\phi_2>0$, then\footnote{In particular, $\int_{c_1^{\ast}}^{c_2^{\ast}}(1-x)f(x)\diff x > 0$} $c_1^{\ast}<c_2^{\ast}$ and $Q$ is the quotient on the \emph{lhs} of \eqref{eqn_app_CBIsoln_continousmarginal_limits_original}. As $n\rightarrow\infty$, integrals of $(1-x)^n$ in $Q$ all tend to $0$ by the \emph{monotone convergence theorem} (m.c.t.) \cite{schilling_2005}. The m.c.t. also implies $\lim\limits_{n\rightarrow\infty}\int_{0}^{b}\frac{(1-2x)^{n-1}}{(1-x)^{n-2}}{\bf 1}_{x\in(c^1_{\ast},c^2_{\ast})}\,f(x)\diff x=0$. Therefore, $\lim\limits_{n\rightarrow\infty}Q=\infty$.
\end{proof}

\renewcommand{\thealgorithm}{} 
\newpage
\section{Algorithm for Numerical Estimates of $c_{\ast}^{1}$, $c_{\ast}^{2}$}
\label{app_D}
For brevity, we omit the analogous algorithm for $c_{1}^{*}$, $c_{2}^{*}$.
\begin{algorithm}
\caption{Bisection Method based Algorithm for $c_{*}^{1}$, $c_{*}^{2}$}
\label{alg_bisec}
\begin{algorithmic}[1]
\Require 
The \textit{pfd} density $f(x)$, an intermediate function
$g_l(x)$, the target \textit{pfd} bound $b$, a tolerance $\epsilon$ and the doubts $\phi_1,\phi_2$.
\Ensure $c_{*}^{1}$, $c_{*}^{2}$ 
\If{$\int_0^b f(u) \diff u>\phi_1$}
    \State $x_l= \argmax g_l(x)$
    \If{$x_l>b$}
    \State $c_{*}^{1}=\mathrm{solve}(\int_x^b f(u) \diff u=\phi_1, x\in[0, b))$ 
    \State $c_{*}^{2}=b$;
    \Return $c_{*}^{1}$, $c_{*}^{2}$
    \Else
    \State $c=\mathrm{solve}(g_l(x)=g_l(b), x\in[0, x_l))$
    \If{$\int_c^b f(u) \diff u<\phi_1$}
    \State $c_{*}^{1}=\mathrm{solve}(\int_x^b f(u) \diff u=\phi_1, x\in[0,b))$ 
    \State $c_{*}^{2}=b$;
    \Return $c_{*}^{1}$, $c_{*}^{2}$
    \Else \Comment{Start of the bisection method}
    \State $c_{*}^{2}=b$
    \State $tmp_{lb}=x_l$
    \State $tmp_{ub}=b$
    \State $tmp_{\phi_1}=\int_c^b f(u) \diff u$
    \While{$\mid tmp_{\phi_1} - \phi_1\mid > \epsilon$}
        \If{$tmp_{\phi_1} > \phi_1$}
           \State $tmp_{ub}=c_{*}^{2}$
           \State $c_{*}^{2}=\frac{c_{*}^{2}+tmp_{lb}}{2}$
        \Else
            \State $tmp_{lb}=c_{*}^{2}$
            \State $c_{*}^{2}=\frac{c_{*}^{2}+tmp_{ub}}{2}$
        \EndIf
         \State $c_{*}^{1}=\! \mathrm{solve}(g_l(x)\!=\! g_l(c_{*}^{2}), x \in [0, x_l))$
         \State $tmp_{\phi_1}=\int_{c_{*}^{1}}^{c_{*}^{2}} f(u) \diff u$
    \EndWhile
    \EndIf \Comment{End of the bisection method}
    \State \Return $c_{*}^{1}$, $c_{*}^{2}$
    \EndIf
\Else \Comment{This is the case when $\int_0^b f(u) \diff u<\phi_1$}
\State print(``PK\ref{cons_reliable_system} violated!'')
\EndIf

\end{algorithmic}
\end{algorithm}

\end{document}